\documentclass[11pt,letterpaper]{article}

\usepackage{helvet}
\usepackage{courier}
\usepackage{fullpage}
\usepackage{environ}
\NewEnviron{box1}[1]{%
	\begin{center}\fbox{\parbox{5in}{%
				{\centering\scshape #1\par}%
				\parskip=1ex
				\everypar{\hangindent=1em}%
				\BODY
}}\end{center}}

\usepackage{bbm}
\usepackage{url}
\usepackage{float,lipsum}
\floatstyle{boxed}
\restylefloat{table}
\restylefloat{figure}
\usepackage{hyperref}

\usepackage{amsmath,amssymb,amstext} % Lots of math symbols and environments
\usepackage{enumerate}
\usepackage[ruled, boxed]{algorithm}
\usepackage{algpseudocode}

\newcommand\Algphase[1]{%
\vspace*{-.7\baselineskip}\Statex\hspace*{\dimexpr-\algorithmicindent-2pt\relax}\rule{\textwidth}{0.4pt}%
\Statex\hspace*{-\algorithmicindent}\textbf{#1}%
\vspace*{-.7\baselineskip}\Statex\hspace*{\dimexpr-\algorithmicindent-2pt\relax}\rule{\textwidth}{0.4pt}%
}

\newcommand\Algphaze[1]{%
	\Statex\hspace*{-\algorithmicindent}\textbf{#1}%
	\vspace*{-.7\baselineskip}\Statex\hspace*{\dimexpr-\algorithmicindent-2pt\relax}\rule{\textwidth}{0.4pt}%
}

\usepackage{verbatim}
\usepackage{subfigure}
\usepackage{color}
\usepackage{latexsym}
\usepackage{comment}

\newcount\Comments
\definecolor{darkgreen}{rgb}{0,0.7,0}
\newcommand{\kibitz}[2]{\ifnum\Comments=1\textcolor{#1}{#2}\fi}

\newcommand{\floor}[1]{\lfloor #1 \rfloor}

\renewcommand{\vec}[1]{\mathbf{#1}}

\newtheorem{definition}{Definition}
\newtheorem{problem}{Problem}

\newtheorem{example}{Example}
\newtheorem{theorem}{Theorem}
\newtheorem{lemma}{Lemma}
\newtheorem{corollary}{Corollary}
\newenvironment{proof}{\paragraph{Proof:}}{\hfill$\square$}

\DeclareMathOperator*{\argmax}{\arg\!\max}

        {\medskip}

\hypersetup{colorlinks=true,linkcolor=blue,citecolor=blue}

\usepackage{environ}

\NewEnviron{problem1}[1]{%
	\begin{center}\fbox{\parbox{6in}{%
				{\centering\scshape #1\par}%
				\parskip=1ex
				\everypar{\hangindent=1em}%
				\BODY
			}}\end{center}}
\begin{document}

% Page heads
%\markboth{G. Zhou et al.}{A Multifrequency MAC Specially Designed for WSN Applications}

% Title portion
\title{Walrasian Pricing in Multi-unit Auctions\footnote{Simina Br\^{a}nzei 
		was supported by the ISF grant 1435/14 administered by the Israeli Academy of Sciences and Israel-USA Bi-national Science Foundation (BSF) grant  2014389 and the I-CORE Program of the Planning and Budgeting Committee and The Israel Science Foundation.
		This project has received funding from the European Research Council (ERC) under the European Union’s
		Horizon 2020 research and innovation programme (grant agreement No 740282). 
		Aris Filos-Ratsikas was supported by the ERC Advanced Grant 321171 (ALGAME). The authors also aknowledge support from
		the Danish National
		Research Foundation and The National Science Foundation of China
		(under the grant
		61361136003) for the Sino-Danish Center for the Theory of Interactive Computation and
		from the Center for Research in Foundations of Electronic Markets (CFEM), supported by
		the Danish Strategic Research Council. A part of this work was done when Simina was visiting the Simons Institute for the Theory of Computing.}}

\author{Simina Br\^anzei \footnote{
Hebrew University of Jerusalem, Israel. E-mail: \textcolor{blue}{\href{mailto:simina.branzei@gmail.com}{simina.branzei@gmail.com}.}}
\and
Aris Filos-Ratsikas \footnote{
University of Oxford, United Kingdom. E-mail: { \textcolor{blue}{\href{mailto:aris.filos-ratsikas@cs.ox.ac.uk}{aris.filos-ratsikas@cs.ox.ac.uk}}}}
\\
\and
Peter Bro Miltersen \footnote{
Aarhus University, Denmark. E-mail: { \textcolor{blue}{\href{mailto:bromille@cs.au.dk}{bromille@cs.au.dk@cs.au.dk}}}}
\and
Yulong Zeng \footnote{
Tsinghua University, China. E-mail: { \textcolor{blue}{\href{mailto:cengyl13@mails.tsinghua.edu.cn}{cengyl13@mails.tsinghua.edu.cn}}}}
}
\date{}

\maketitle

\begin{abstract}
Multi-unit auctions are a paradigmatic model, where a seller brings multiple units of a good,
while several buyers bring monetary endowments. It is well known that Walrasian equilibria do not always exist in this model, however compelling relaxations such as Walrasian envy-free pricing do. In this paper we design an optimal envy-free mechanism for multi-unit auctions with budgets. When the market
is even mildly competitive, the approximation ratios of this mechanism
are small constants for both the revenue and welfare objectives, and in
fact for welfare the approximation converges to 1 as the market becomes
fully competitive. We also give an impossibility theorem, showing that
truthfulness requires discarding resources and, in particular, is incompatible with (Pareto) efficiency.
\end{abstract}

%\newpage

\section{Introduction} \label{sec:intro}

Auctions are procedures for allocating goods that have been studied in economics in the 20th century, and which are even more relevant now due to the 
emergence of online platforms. Major companies such as Google and Facebook make most of their revenue through auctions, while an increasing number of governments around the world use spectrum auctions to allocate licenses for electromagnetic spectrum to companies. These transactions involve hundreds or thousands of participants with complex preferences, reason for which auctions require more careful design and their study has resurfaced in the computer science literature. 

In this paper we study a paradigmatic model known as multi-unit auctions with budgets, 
in which a seller brings multiple units of a good (e.g. apples), while the buyers bring money and have interests in consuming the goods. Multi-unit auctions have been studied in a large body of literature due to the importance of the model, which already illustrates complex phenomena \cite{dobzinski2012multi,borgs2005multi,dobzinski2007mechanisms,dobzinski2014efficiency,dobzinski2015multi}.

The main requirements from a good auction mechanism are usually computational efficiency, revenue maximization for the seller, and simplicity of use for the participants, the latter of which is captured through the notion of truthfulness. An important property that is often missing from auction design is fairness, and in fact for the purpose of maximizing revenue it is useful to impose higher payments to the buyers that are more interested in the goods. However, there are studies showing that customers are unhappy with such discriminatory prices (see, e.g., \cite{anderson2008}), which has lead to a body of literature focused on achieving fair pricing \cite{guruswami2005profit,feldman2012revenue,cohen2016invisible,feldman2015combinatorial,sekar2016posted}.

%This literature can be at a high level classified in two strands, on one hand studying the problem of designing incentive-compatible auction mechanisms optimized for revenue with little regard for fairness, and on the other hand analyzing Walrasian equilibria (or its relaxations) with no guarantees about its performance when the participants act based on incentives.

A remarkable solution concept that has been used for achieving fairness in auctions comes from free markets, which are economic systems where the prices and allocations are not designed by a central authority. Instead, the prices emerge through a process of adjusting demand and supply such that everyone faces the same prices and the buyers freely purchase the bundles they are most interested in. When the goods are divisible, an outcome where supply and demand are perfectly balanced---known as competitive (or Walrasian) equilibrium \cite{Walras74} ---always exists under mild assumptions on the utilities and has the property that the participants face the same prices and can freely acquire their favorite bundle at those prices. The competitive equilibrium models outcomes of large economies, where the goods are divisible and the participants  so small (infinitesimal) that they have no influence on the market beyond purchasing their most preferred bundle at the current prices. Unfortunately, when the goods are indivisible, the competitive equilibrium does not necessarily exist (except for small classes of valuations see, e.g., \cite{KC82,GS99}) and the induced mechanism -- the Walrasian mechanism \cite{babaioff2014efficiency,cheung2008approximation} -- is generally manipulable.

%In bipartite market models, which are relevant for auctions, the participants are categorized as sellers---who bring goods for sale---and buyers---who bring money and have interests in consuming the goods.

%In this work we study the problem of designing an auction mechanism that is fair (as the Walrasian equilibrium solution concept), as well as 
%unifies these goals, in a way that guarantees good revenue for the seller and fairness and high welfare for the participants.
A solution for recovering the attractive properties of the Walrasian equilibrium in the multi-unit model is to relax the clearing requirement of the market equilibrium, by allowing the seller to not sell all of the units. This solution is known as (Walrasian) envy-free pricing \cite{guruswami2005profit}, and it ensures that all the participants of the market face the same prices\footnote{The term envy-free pricing has also been used when the pricing is per-bundle, not per-item. We adopt the original definition of \cite{guruswami2005profit} which applies to unit-pricing, due to its attractive fairness properties \cite{feldman2012revenue}.}, and each one purchases their favorite bundle of goods. An envy-free pricing trivially exists by pricing the goods infinitely high, so the challenge is finding one with good guarantees, such as high revenue for the seller or high welfare for the participants. 

We would like to obtain envy-free pricing mechanisms that work well with strategic participants, who may alter their inputs to the mechanism to get better outcomes. To this end, we design an optimal truthful and envy-free mechanism for multi-unit auctions with budgets, with high revenue and welfare in competitive environments. Our work can be viewed as part of a general research agenda of \emph{simplicity in mechanism design} \cite{hartline2009simple}, which recently proposed item pricing \cite{balcan2008item,feldman2015combinatorial} as a way of designing simpler auctions while at the same time avoiding the ill effects of discriminatory pricing \cite{feldman2012revenue,anderson2008}. Item pricing is used in practice all over the world to sell goods in supermarkets or online platforms such as Amazon, which provides a strong motivation for understanding it theoretically.

\subsection{Model and Results}

Our model is a multi-unit auction with budgets, in which a seller owns $m$ identical units of an item. Each buyer $i$ has a budget $B_i$ and a value $v_i$ per unit. The utilities of the buyers are quasi-linear up to the budget cap, while any allocation that exceeds that cap is unfeasible. %, which is denoted by utility $-\infty$. 

We deal with the problem of designing envy-free pricing schemes for the strongest concept of incentive compatibility, namely dominant strategy truthfulness. The truthful mechanisms are in the \emph{prior-free} setting, i.e. they do not require any prior distribution assumptions. We evaluate the efficiency of mechanisms using the notion of \emph{market share}, $s^*$, which captures the maximum buying power of any individual buyer in the market. Our main contributions can be summarized as follows.
\vspace{2mm}

\noindent \textbf{Main Theorem} (informal) \vspace{0.1mm}\emph{For linear multi-unit auctions with known monetary endowments: 
	\hspace{4mm}	
	\begin{itemize}
		\item There exists no (Walrasian) envy-free mechanism that is both truthful and non-wasteful.
		\item There exists a truthful (Walrasian) envy-free auction, which attains a fraction of at least $\max\left\{2, \frac{1}{1-s^*}\right\}$ of the optimal revenue and at least $1-s^*$ of the optimal welfare on any market, where $0 < s^* < 1$ is the market share. This mechanism is optimal for both the revenue and welfare objectives when the market is even mildly competitive (i.e. with market share $s^* \leq 50\%$), and its approximation for welfare converges to $1$ as the market becomes fully competitive. 
	\end{itemize} 
}
\medskip

In the statement above, optimal means that there is no other truthful envy-free auction mechanism with a better approximation ratio. A mechanism is non-wasteful if it allocates as many units as possible at a given price.
The impossibility theorem implies in particular that truthfulness is incompatible with Pareto efficiency.
Our positive results are for \emph{known} budgets, similarly to \cite{dobzinski2012multi}. In the economics literature budgets are viewed as hard information (quantitative), as opposed to the valuations, which represent soft information and are more difficult to verify (see, e.g., \cite{Petersen}).

\bigskip

We also provide several computational results: a polynomial time algorithm for computing a welfare maximizing envy-free pricing, and an FPTAS and an exact algorithm (which runs in polynomial-time for a constant number of types of buyers) for computing a revenue-maximizing envy-free pricing (Theorem \ref{thm:max_revenue} and \ref{thm:max_revenue_constant} in the appendix). Our FPTAS for revenue improves upon the results in~\cite{feldman2012revenue}, which had previously provided a 2-approximation algorithm. 

Finally, in the general multi-unit model, we show hardness of maximizing welfare and revenue and provide an FPTAS 
for both objectives (Theorem \ref{thm:general_hard} and \ref{thm:general_fptas} in the appendix).

\subsection{Related Work}

The multi-unit setting has been studied in a large body of literature on auctions (\cite{dobzinski2012multi,borgs2005multi,dobzinski2007mechanisms,dobzinski2014efficiency,dobzinski2015multi}), where the focus has been on designing truthful auctions with good approximations to some desired objective, such as the social welfare or the revenue. Quite relevant to ours is the paper by \cite{dobzinski2012multi}, in which the authors study multi-unit auctions with budgets, however with no restriction to envy-free pricing or even item-pricing. They design a truthful auction  (that uses discriminatory pricing) for \emph{known budgets}, that achieves near-optimal revenue guarantees when the influence of each buyer in the auction is bounded, using a notion of buyer \emph{dominance}, which is conceptually close to the market share notion that we employ. Their mechanism is based on the concept of clinching auctions \cite{ausubel2004efficient}.

Attempts at good prior-free truthful mechanisms for multi-unit auctions are seemingly impaired by their general impossibility result which states that truthfulness and efficiency are essentially incompatible when the budgets are \emph{private}. Our general impossibility result is very similar in nature, but it is not implied by the results in \cite{dobzinski2012multi} for the following two reasons: (a) our impossibility holds for \emph{known} budgets and (b) our notion of efficiency is weaker, as it is naturally defined with respect to envy-free allocations only. This also means that our impossibility theorem is not implied by their uniqueness result, even for two buyers.
Multi-unit auctions with budgets have also been considered in \cite{dobzinski2014efficiency} and \cite{borgs2005multi}, and without budgets (\cite{dobzinski2015multi,bartal2003incentive,dobzinski2007mechanisms}); all of the aforementioned papers do not consider the envy-freeness constraint.

The effects of strategizing in markets have been studied extensively over the past few years (\cite{borodin2015budgetary,branzei2014fisher,chen2011profitable,mehta2013exchange,mehta2014save}). For more general envy-free auctions, besides the multi-unit case, there has been some work on truthful mechanisms in the literature of envy-free auctions (\cite{guruswami2005profit}) and (\cite{HY11}) for \emph{pair envy-freeness}, a different notion which dictates that no buyer would want to swap its allocation with that of any other buyer \cite{mt16}. It is worth noticing that there is a body of literature that considers envy-free pricing as a purely optimization problem (with no regard to incentives) and provides approximation algorithms and hardness results for maximizing revenue and welfare in different auction settings \cite{feldman2012revenue,colini2014revenue}.

It is worth mentioning that the good approximations achieved by our truthful mechanism are a \emph{prior-free setting} (\cite{hartline2013mechanism}),  i.e. we don't require any assumptions on prior distributions from which the input valuations are drawn. Good prior-free approximations are usually difficult to achieve and a large part of the literature is concerned with auctions under distributional assumptions, under the umbrella of \emph{Bayesian mechanism design} (\cite{DM17,Correa17,DHP17,cai2012optimal,cai2013reducing,hartline2013mechanism,myerson1981optimal}).

\section{Preliminaries}

In a linear multi-unit auction with budgets there is a set of buyers, denoted by $N = \{1, \ldots, n\}$, and a single seller with $m$ indivisible units of a good for sale.
Each buyer $i$ has a valuation $v_i > 0$ and a budget $B_i > 0$, both drawn from a discrete domain $\mathbb{V}$ of rational numbers: $v_i, B_i \in \mathbb{V}$. 
The valuation $v_i$ indicates the value of the buyer for one unit of the good.

An \emph{allocation} is an assignment of units to the buyers denoted by a vector $\vec{x} = (x_1, \ldots, x_n)$ $\in \mathbb{Z}^{n}_{+}$, where $x_i$ is the number of units received by buyer $i$. We are interested in feasible allocations, for which: $\sum_{i=1}^{n} x_i \leq m$.

The seller will set a price $p$ per unit, such that the price of purchasing $\ell$ units is $p \cdot \ell$ for any buyer. The interests of the buyers at a given price are captured by the demand function.

\begin{definition}[Demand] 
	The \emph{demand} of buyer $i$ at a price $p$ is a set consisting of all the possible bundle sizes (number of units) that the buyer would like to purchase at this price:
	
	$$D_i(p)=\begin{cases}
	\min\{\lfloor \frac{B_i}{p} \rfloor,m\}, & \text{if $p < v_i$}\\[4pt]
	0,\ldots,\min\{\lfloor\frac{B_i}{p}\rfloor ,m\},  & \text{if $p = v_i$}\\[4pt]
	\ \ 0, & \text{otherwise}.
	\end{cases}$$
	
	If a buyer is indifferent between buying and not buying at a price, then its demand is a set of all the possible bundles that it can afford, based on its budget constraint.
\end{definition}

\begin{definition}[Utility]
	The \emph{utility} of buyer $i$ given a price $p$ and an allocation $\vec{x}$ is 
	$$u_i(p, x_i)=\begin{cases}
	v_i \cdot x_i - p \cdot x_i, & \text{if $p \cdot x_i \leq B_i$}\\
	- \infty, & \text{otherwise}
	\end{cases}$$
\end{definition}

\noindent \textbf{(Walrasian) Envy-free Pricing}. An allocation and price $(\vec{x},p)$ represent a (Walrasian) \emph{envy-free pricing} if each buyer is allocated a number of units in its demand set at price $p$, i.e. $x_i \in D_i(p)$ for all $i \in N$. A price $p$ is an \emph{envy-free price} if there exists an allocation $\mathbf{x}$ such that $(\vec{x},p)$ is an envy-free pricing. 

While an envy-free pricing always exists (just set $p = \infty$), it is not always possible to sell \emph{all} the units in an envy-free way. We illustrate this through an example.

\begin{example}[Non-existence of envy-free clearing prices]
	Let $N = \{1, 2\}$, $m= 3$, valuations $v_1 = v_2 = 1.1$, and $B_1 = B_2 = 1$. At any price $p > 0.5$, no more than $2$ units can be sold in total because of budget constraints. At $p \leq 0.5$, both buyers are interested and demand at least $2$ units each, but there are only $3$ units in total.
\end{example}

\noindent \textbf{Objectives}. We are interested in maximizing the \emph{social welfare} and \emph{revenue} objectives attained at envy-free pricing. The \emph{social welfare} at an envy-free pricing $(\vec{x}, p)$ is the total value of the buyers for the goods allocated, while the \emph{revenue} is the total payment received by the seller, i.e.
$
\mathcal{SW}(\vec{x}, p) = \sum_{i=1}^n v_{i}\cdot x_i$ and $\mathcal{REV}(\vec{x}, p) = \sum_{i=1}^n x_i \cdot p.
$ \\

\noindent \textbf{Mechanisms}. The goal of the seller will be to obtain money in exchange for the goods, however, it can only do that if the buyers are interested in purchasing them. The problem of the seller will be to obtain accurate information about the preferences of the buyers that would allow optimizing the pricing.
Since the inputs (valuations) of the buyers are private, we will aim to design auction mechanisms that incentivize the buyers to reveal their true preferences~\cite{AGT_book}.

An auction \emph{mechanism} is a function $M: \mathbb{V}^n \rightarrow \mathbb{O} \times \mathbb{Z}_{+}^n$ that maps the valuations reported by the buyers to a price $p \in \mathbb{O}$, where $\mathbb{O}$ is the space from which the prices are drawn\footnote{In principle the spaces $\mathbb{V}$ and $\mathbb{O}$ can be the same but for the purpose of getting good revenue and welfare, it is useful to have the price to be drawn from a slightly larger domain; see Section \ref{sec:AON}.}, and an allocation vector $\vec{x} \in \mathbb{Z}_{+}^n$. 

\begin{definition}[Truthful Mechanism]
	A mechanism $M$ is \emph{truthful} if it incentivizes the buyers to reveal their true inputs, i.e. $u_i(M(\vec{v})) \geq u_i(M(v_i',v_{-i}))$, for all $i \in N$, any alternative report $v_i' \in \mathbb{V}$ of buyer $i$ and any vector of reports $v_{-i}$ of all the other buyers.
\end{definition}

Requiring incentive compatibility from a mechanism can lead to worse revenue, so our goal will be to design mechanisms that achieve revenue close to that attained in the pure optimization problem (of finding a revenue optimal envy-free pricing without incentive constraints).

\medskip 

\noindent \textbf{Types of Buyers}. The next definitions will be used extensively in the paper. Buyer $i$ is said to be \emph{hungry} at price $p$ if $v_i > p$ and \emph{semi-hungry} 
if $v_i = p$. Given an allocation $\vec{x}$ and a price $p$ buyer $i$ is \emph{essentially hungry} if it is either semi-hungry with $x_i = \min\{\lfloor B_i/p \rfloor,m\}$ or hungry. In other words, a buyer is essentially hungry if its value per unit is at least as high as the price per unit and, moreover, the buyer receives the  largest non-zero element in its demand set. 

\medskip

The following lemmas will be useful.

\begin{lemma} \label{lem:small_to_large}
	Let $p$ be an envy-free price. Then any price $p' > p$ is also envy-free. Similarly, if $p$ is not an envy-free price, then any price $p' < p$ is not envy-free either.
\end{lemma}
\begin{proof}
	This follows from the fact that for every buyer $i$, the number of demanded units is non-increasing in the price. If at price $p$ there are enough units to satisfy all demands, then the same holds at any price $p'>p$. Similarly, if at some price $p$ there are not enough units to satisfy all demands, this is also the case for any $p'<p$. 
\end{proof}

For both revenue and welfare, the optimal solution can be found in a set of \emph{candidate prices}: 
$$\mathcal{P} = \left\{v_{i}, \frac{B_i}{k} \; | \; \forall i \in N, \forall k \in [m]\right\}.$$ These prices are either equal to some valuation or have the property that some buyer could exhaust its budget by purchasing all the units it can afford. 

\begin{lemma}\label{lem:inthesetP}
	For both the revenue and social welfare objectives, there is an optimal envy-free price $p \in \mathcal{P}$.
\end{lemma}
\begin{proof}
	Let $p'$ be a welfare maximizing envy-free price and let $\vec{x}$ be the corresponding allocation. If $p' \in \mathcal{P}$ then we are done. Else, assume $p' \neq \mathcal{P}$. Then, we can increase the price $p'$ until some budget $B_i$ is exhausted or the price becomes equal to some valuation $v_i$. Until that happens, the demand sets of all buyers remain constant and hence the exact same allocation $\vec{x}$ can be supported at some price $p \in \mathcal{P}$. Since the social welfare only depends on the allocation and not the price, the conclusion follows.
	The proof for revenue follows from the observation that increasing the price is beneficial for the seller as long as it continues to sell the same number of items. The discontinuities only happen at points where the price matches the valuation at some buyer (and so increasing the price above that value can result in losing the buyer) or when the number of items decreases because a buyer can no longer afford to purchase as many units.
\end{proof}

\begin{lemma}\label{lem:frompricetoalloc}
	For a linear multi-unit market, given an envy-free price $p$, a revenue or welfare maximizing allocation at $p$ can be found in polynomial time in $n$ and $\log(m)$.
\end{lemma}
\begin{proof}
	First, given the valuation functions of the hungry buyers, we can compute their demands at price $p$. Note these demands are singletons and so the allocation for these buyers is uniquely determined. For the non-hungry buyers (if any), we assign the remaining units (if any) in a greedy fashion: Fix an arbitrary order of buyers and assign them units according to that order, until all of them exhaust their budgets or we run out of units. All these operations can be done in polynomial time.	
\end{proof}

\section{An optimal envy-free and truthful mechanism}\label{sec:AON}

In this section, we present our main contribution, an envy-free and truthful mechanism, which is optimal among all truthful mechanisms and achieves small constant approximations to the optimal welfare and revenue. The approximation guarantees are with respect to the \emph{market-share} $s^*$, which intuitively captures the maximum purchasing power of any individual buyer in the auction. The formal definition is postponed to the corresponding subsection.

\begin{theorem}\label{thm:AON-ar-sw}
	There exists a truthful (Walrasian) envy-free auction, which attains a fraction of at least
	\begin{itemize}
		\setlength{\itemindent}{.2in}
		\item $\max\left\{2, \frac{1}{1-s^*}\right\}$ of the optimal revenue, and
		\item $1-s^*$ of the optimal welfare
	\end{itemize}
	on any market. 	
	This mechanism is optimal for both the revenue and welfare objectives when the market is even mildly competitive (i.e. with market share $s^* \leq 50\%$), and its approximation for welfare converges to $1$ as the auction becomes fully competitive.
\end{theorem}

\noindent Consider the following mechanism.

\begin{box1}{\emph{\textbf{\textsc{All-or-Nothing:}}}}
	Given as input the valuations of the buyers, let $p$ be the minimum envy-free price and
	$\vec{x}$ the allocation obtained as follows: 
	\begin{itemize}
		\item For every hungry buyer $i$, set $x_i$ to its demand. \\[-5pt]
		\item For every buyer $i$ with $v_i<p$, set $x_i = 0$.  \\ [-5pt]
		\item For every semi-hungry buyer $i$, set $x_i = \lfloor B_i / p \rfloor$ if possible, otherwise set $x_i=0$ taking the semi-hungry buyers in lexicographic order.
	\end{itemize}
\end{box1} 

\bigskip

\noindent In other words, the mechanism always outputs the minimum envy-free price but if there are semi-hungry buyers at that price, they get either all the units they can afford at this price or $0$, even if there are still available units, after satisfying the demands of the hungry buyers. \\

\begin{lemma}
	The minimum envy-free price does not exist when the price domain is $\mathbb{R}$.
\end{lemma}
\begin{proof}
	If the price can be any real number, consider an auction with $n=2$ buyers, $m=2$ units, valuations $v_1 = v_2 = 3$ and budgets $B_1 = B_2 = 2$.
	At any price $p \leq 1$, there is overdemand since each buyer is hungry and demands at least $2$ units, while there are only $2$ units in total. At any price $p \in (1, 2]$, each buyer demands at most one unit due to budget constraints, and so all the prices in the range $(1,2]$ are envy-free. This is an open set, and so there is no minimum envy-free price. Note however, that by making the output domain discrete, e.g. with $0.1$ increments starting from zero, then the minimum envy-free price output is $1.01$. At this price each buyer purchases $1$ unit.
\end{proof}

Given the example above, we will consider the discrete domain $\mathbb{V}$ as an infinite grid with entries of the form $k\cdot \epsilon$, for $k \in \mathbb{N}$ and some sufficiently small\footnote{For most of our results, any discrete domain is sufficient for the results to hold; for some results we will need to a number of grid points that polynomial in the size of the input grid.}  $\epsilon$. For the output of the mechanism, we will assume a slightly finer grid, e.g. with entries $k\cdot \delta=k(\epsilon/2)$, for $k \in \mathbb{N}$. The minimum envy-free price can be found in time which is polynomial in the input and $\log(1/\epsilon)$, using binary search\footnote{In the full version, we describe a faster procedure that finds the minimum envy-free without requiring to do binary search over the grid.}  and the mechanism is optimal with respect to discrete domain that we operate on. Operating on a grid is actually without loss of generality in terms of the objectives; even if we compare to the optimal on the continuous domain, if our discretization is fine enough, we don't lose any revenue or welfare.

\begin{lemma}\label{thm:continuous}
	When the valuation and budget of each buyer are drawn from a discrete grid with entries $k\cdot \epsilon$, and the price is is drawn from a finer grid with entries $k\cdot \epsilon/2$, for $k \in \mathbb{N}$, then the welfare and revenue loss of the \textsc{All-or-Nothing} mechanism due to the discretization of the output domain is zero. The mechanism always runs in time polynomial in the input and $\log(1/\epsilon)$.
\end{lemma}

%\begin{theorem}\label{thm:continuous}
%	When the valuation and budget of each buyer are rational numbers, where both the numerator and denominator are integers specified with $t$ bits of precision, and the price is specified with $s \geq 4t + 4\log{m}$ bits, then the welfare and revenue loss of the \textsc{All-or-Nothing} mechanism due to the discretization of the output domain is zero. The mechanism always runs in $O(t+s)$.
%\end{theorem}

%Thus by setting $s = 4t + 4 \log{m}$, we simultaneously obtain a runtime that is polynomial in the input, $O(t + \log{m})$, and that there is no loss from restricting \textsc{All-or-Nothing} to output a price on the grid, compared to the revenue and welfare obtained  when the price can be an arbitrary real number (or the supremum value if the minimum envy-free price does not exist).

\subsection*{Truthfulness of the \textsc{All-or-Nothing} Mechanism}

\noindent The following theorem establishes the truthfulness of \textsc{All-or-Nothing}. 

\begin{theorem} \label{thm:sp_aon}
	The \textsc{All-or-Nothing} mechanism is truthful.
\end{theorem}

\begin{proof}
	First, we will prove the following statement. If $p$ is any envy-free price and $p'$ is an envy-free price such that $p\leq p'$ then the utility of any essentially hungry buyer $i$ at price $p$ is at least as large as its utility at price $p'$. The case when $p'=p$ is trivial, since the price (and the allocation) do not change. Consider the case when $p<p'$. Since $p$ is an envy-free price, buyer $i$ receives the maximum number of items in its demand. For a higher price $p'$, its demand will be at most as large as its demand at price $p$ and hence its utility at $p'$ will be at most as large as its utility at $p$.  
	
	Assume now for contradiction that Mechanism \textsc{All-or-Nothing} is not truthful and let $i$ be a deviating buyer who benefits by misreporting its valuation $v_i$ as $v_i'$ at some valuation profile $\mathbf{v}=(v_1,\ldots, v_n)$, for which the minimum envy-free price is $p$. Let $p'$ be the new minimum envy free price and let $\mathbf{x}$ and $\mathbf{x'}$ be the corresponding allocations at $p$ and $p'$ respectively, according to \textsc{All-or-Nothing}. Let $\mathbf{v'}=(v_i',v_{-i})$ be the valuation profile after the deviation.
	
	We start by arguing that the deviating buyer $i$ is essentially hungry. First, assume for contradiction that $i$ is neither hungry nor semi-hungry, which means that $v_i <p$. Clearly, if $p'\geq p$, then buyer $i$ does not receive any units at $p'$ and there is no incentive for manipulation; thus we must have that $p'<p$. This implies that every buyer $j$ such that $x_j>0$ at price $p$ is hungry at price $p'$ and hence $x_j' \geq x_j$. Since the demand of all players does not decrease at $p'$, this implies that $p'$ is also an envy-free price on instance $\mathbf{v}$, contradicting minimality of $p$. 
	
	Next, assume that buyer $i$ is semi-hungry but not essentially hungry, which means that $v_i=p$ and $x_i =0$, by the allocation of the mechanism. Again, in order for the buyer to benefit, it has to hold that  $p'<p$ and $x_i' >0$ which implies that $x_i' = \lfloor B_i/p'\ \rfloor$, i.e. buyer $i$ receives the largest element in its demand set at price $p'$. But then, since $p' <p$ and $p'$ is an envy-free price, buyer $i$ could receive $\lfloor B_i/p \rfloor$ units at price $p$ without violating the envy-freeness of $p$, in contradiction with each buyer $i$ being essentially hungry at $p$.
	
	From the previous two paragraphs, the deviating buyer must be essentially hungry. This means that $x_i>0$ and $v_i \geq p$. By the discussion in the first paragraph of the proof, we have $p'<p$. 
	Since $x_i>0$, the buyer does not benefit from reporting $v_i'$ such that $v_i' < p'$. Thus it suffices to consider the case when $v_i' \geq p'$. We have two subcases: \\
	\begin{itemize}
		\item $v_i'>p$: Buyer $i$ is essentially hungry at price $p$ according to $v_i$ and hungry at price $p'$ according to $v_i'$. The reports of the other buyers are fixed and $B_i$ is known; similarly to above, price $p'$ is an envy-free price on instance $\mathbf{v}$, contradicting the minimality of $p$. \\
		
		\item $v_i'=p'$: Intuitively, an essentially hungry buyer at price $p$ is misreporting its valuation as being lower trying to achieve an envy-free price $p'$ equal to the reported valuation. 
		Since $v_i'=p'$,  Mechanism \textsc{All-or-Nothing} gives the buyer either as many units as it can afford at this price or zero units. In the first case, since $p'$ is envy-free and $B_i$ is known, buyer $i$ at price $p'$ receives the largest element in its demand set %at price $p$ (according to its real value $v_i$) 
		and since the valuations of all other buyers are fixed, $p'$ is also an envy-free price on input $\mathbf{v}$, contradicting the minimality of $p$. In the second case, the buyer does not receive any units and hence it does not benefit from misreporting.  \\
	\end{itemize}
	Thus there are no improving deviations, which concludes the proof of the theorem.
\end{proof} 

\subsection*{Performance of the \textsc{All-or-Nothing} Mechanism}

Next, we show that the mechanism has a good performance for both objectives. 
We measure the performance of a truthful mechanism by the standard notion of approximation ratio, i.e. 
$$\text{ratio(M)} = \sup_{\mathbf{\vec{v} \in \mathbb{R}^n}} \frac{\max_{\mathbf{x},p} \mathcal{OBJ}(\vec{v})}{\mathcal{OBJ}(M(\vec{v}))},$$
\noindent where $\mathcal{OBJ} \in \{\mathcal{SW,REV}\}$ is either the social welfare or the revenue objective. Obviously, a mechanism that outputs a pair that maximizes the objectives has approximation ratio $1$. The goal is to construct truthful mechanisms with approximation ratio as close to $1$ as possible.

We remark here that for the approximation ratios, we only need to consider valuation profiles that are not ``trivial'', i.e. input profiles for which at any envy-free price, no hungry or semi-hungry buyers can afford a single unit and hence the envy-free price can be anything;
on trivial profiles, both the optimal price and allocation and the price and allocation output by Mechanism \textsc{All-or-Nothing} obtain zero social welfare or zero revenue. \\

\noindent \textbf{Market Share} A well-known notion for measuring the competitiveness of a market is the \emph{market share}, understood as the percentage of the market accounted for by a specific entity (see, e.g., \cite{farris2010marketing}, Chapter 2).

In our model, the maximum purchasing power (i.e. number of units) of any buyer in the auction occurs at the minimum envy-free price, $p_{min}$. By the definition of the demand, there are many ways of allocating the semi-hungry buyers, so when measuring the purchasing power of an individual buyer we consider the maximum number of units that buyer can receive, taken over the set of all feasible maximal allocations at $p_{min}$. Let this set be $\mathcal{X}$.
Then the market share of buyer $i$ can be defined as:
$$s_i = \max_{\vec{x} \in \cal{X}} \left(\frac{x_i}{\sum_{k=1}^n x_k}\right).$$

Then, the \emph{market share} is defined as 
$s^* = \max_{i=1}^{n} s_i$. Roughly speaking, a market share $s^* \leq 1/2$ means that a buyer can never purchase more than half of the resources.

\begin{theorem} \label{thm:revenuebound}
	The \textsc{All-or-Nothing} mechanism approximates the optimal revenue within a factor of $2$ whenever the market share, $s^*$, is at most $50\%$.
\end{theorem}
\begin{proof}	
	Let $OPT$ be the optimal revenue, attained at some price $p^*$ and allocation $\vec{x}$, and $\mathcal{REV}(AON)$ the revenue attained by the $\textsc{All-or-Nothing}$ mechanism. By definition, mechanism $\textsc{All-or-Nothing}$ outputs the minimum envy-free price $p_{min}$, together with an allocation $\vec{z}$. 
	For ease of exposition, let $\alpha_i = B_i/p_{min}$ and $\alpha_i^* = B_i/p^*$, $\forall i \in N$. 
	There are two cases, depending on whether the optimal envy-free price, $p^*$, is equal to the minimum envy-free price, $p_{min}$:\\
	
	\noindent \emph{Case 1}: $p^* > p_{min}$.
	Denote by $L$ the set of buyers with valuations at least $p^*$ that can afford at least one unit at the optimal price.
	Note that the set of buyers that get allocated at $p_{min}$ is a superset of $L$. Moreover, the optimal revenue is bounded by the revenue attained at the (possibly infeasible) allocation where all the buyers in $L$ get the maximum number of units in their demand.
	These observations give the next inequalities:
	$$\mathcal{REV}(AON) \geq \sum_{i \in L} \left \lfloor \alpha_i \right \rfloor \cdot p_{min} \ \ \text{and} \ \ 
	OPT \leq \sum_{i \in L} \left \lfloor \alpha_i^* \right \rfloor \cdot p^*.$$
	Then the revenue is bounded by:
	\begin{align*}
	\frac{\mathcal{REV}(AON)}{OPT}  &\geq 
	\frac{\sum_{i \in L} \left \lfloor \alpha_i \right \rfloor \cdot p_{min}}{\sum_{i \in L} \left \lfloor \alpha_i^* \right \rfloor \cdot p^*} 
	\geq  
	\frac{\sum_{i \in L} \left \lfloor \alpha_i  \right \rfloor  \cdot p_{min}}{\sum_{i \in L} \alpha_i^* \cdot p^*}
	= \frac{\sum_{i \in L} \left \lfloor \alpha_i \right \rfloor  \cdot p_{min}}{\sum_{i \in L} B_i}\\
	& = 
	\frac{\sum_{i \in L} \left \lfloor \alpha_i \right \rfloor }{\sum_{i \in L} \alpha_i } 
	\geq 
	\frac{\sum_{i \in L} \left \lfloor \alpha_i \right \rfloor}{\sum_{i \in L} 2 \left \lfloor \alpha_i \right \rfloor} 
	= \frac{1}{2},
	\end{align*}
	\noindent where we used that the auction is non-trivial, i.e. for any buyer $i \in L$, $\left \lfloor \alpha_i \right \rfloor \geq 1$, and so 
	$\alpha_i  \leq \left \lfloor \alpha_i \right \rfloor + 1 \leq 2 \left \lfloor \alpha_i \right \rfloor$. \\
	
	\noindent \emph{Case 2}: $p^* = p_{min}$. The hungry buyers at $p_{min}$, as well as the buyers with valuations below $p_{min}$, receive identical allocations under $\textsc{All-or-Nothing}$ and the optimal allocation, $\vec{x}$. However there are multiple ways of assigning the semi-hungry buyers to achieve an optimal allocation.
	Recall that $\vec{z}$ is the allocation made by $\textsc{All-or-Nothing}$. Without loss of generality, we can assume that $\vec{x}$ is an optimal allocation with the property that $\vec{x}$ is a superset of $\vec{z}$ and the following condition holds: \\
	\begin{itemize}
		\item \emph{the number of buyers not allocated under $\vec{z}$, but that are allocated under $\vec{x}$, is minimized.} \\
	\end{itemize} 
	We argue that $\vec{x}$ allocates at most one buyer more compared to $\vec{z}$. 
	Assume by contradiction that there are at least two semi-hungry buyers $i$ and $j$, such that 
	$0 < x_i < \left \lfloor \alpha_i \right \rfloor$ and $0 < x_j < \left \lfloor \alpha_j \right \rfloor$. Then we can progressively take units from buyer $j$ and transfer them to buyer $i$, until either buyer $i$ receives 
	$x_i' = \left \lfloor \alpha_i \right \rfloor$, or buyer $j$ receives $x_j' = 0$. 
	Hence we can assume that the set of semi-hungry buyers that receive non-zero, non-maximal allocations in the optimal solution $\vec{x}$ is either empty or a singleton.
	If the set is empty, then $\textsc{All-or-Nothing}$ is optimal. Otherwise, let the singleton be $\ell$; denote by $\tilde{x}_{\ell}$ the maximum number of units that $\ell$ can receive in any envy-free allocation at $p_{min}$. Since the number of units allocated by any maximal envy-free allocation at $p_{min}$ is equal to $\sum_{i=1}^n x_i$, but $x_{\ell} \leq \tilde{x}_{\ell}$, we get: 
	$$\frac{x_{\ell}}{\sum_{i=1}^n x_i} \leq \frac{\tilde{x}_{\ell}}{\sum_{i=1}^n x_i} = s_i^*.$$
	Thus
	\begin{align*}
	\frac{\mathcal{REV}(AON)}{OPT} & = \frac{OPT - x_{\ell} \cdot p_{min}}{OPT} \geq 
	\frac{OPT - \tilde{x}_{\ell} \cdot p_{min}}{OPT} 
	= 1 - \frac{\tilde{x}_{\ell} \cdot p_{min}}{\sum_{i=1}^n x_i \cdot p_{min}} \\
	&=
	1 - \frac{\tilde{x}_{\ell}}{\sum_{i=1}^{n} x_i} 
	= 1 - s_i^{*} \geq 1 - s^*
	\end{align*}
	Combining the two cases, the bound follows. This completes the proof.
\end{proof} 

\begin{corollary}
	The performance of the \textsc{All-or-Nothing} mechanism is $\max\{2,1/(1-s^*)$ on any market (i.e. with market share $0 < s^* < 1$).
\end{corollary}
\begin{proof}
	From the proof of Theorem \ref{thm:revenuebound}, since the arguments of Case 1 do not use the market share $s^*$, it follows that the ratio of \textsc{All-Or-Nothing} for the revenue objective can alternatively be stated as $\max\{2,1/(1-s^*)\}$ and therefore it degrades gracefully with the increase in the market share.
\end{proof}

\noindent The next theorem establishes that the approximation ratio for welfare is also constant. 

\begin{theorem}\label{thm:AON-ar-socialwelfare}
	The approximation ratio of Mechanism \textsc{All-or-Nothing} with respect to the social welfare is at most $1/(1-s^*)$, where the market share $s^* \in (0,1)$. The approximation ratio goes to $1$ as the market becomes fully competitive.
\end{theorem}
\begin{proof}
	For social welfare we have, similarly to Theorem \ref{thm:revenuebound}, that
	\begin{align*}
	\frac{\mathcal{SW}(AON)}{OPT} & = \frac{OPT - x_{\ell} \cdot v_{\ell}}{OPT} \geq 
	\frac{OPT - \tilde{x}_{\ell} \cdot v_{\ell}}{OPT} 
	= 1 - \frac{\tilde{x}_{\ell} \cdot v_{\ell}}{\sum_{i=1}^n x_i \cdot v_{i}}
	\geq 1 - \frac{\tilde{x}_{\ell} \cdot v_{\ell}}{\sum_{i=1}^n x_i \cdot v_{\ell}} \\
	& =
	1 - \frac{\tilde{x}_{\ell}}{\sum_{i=1}^{n} x_i} 
	= 1 - s_i^{*} \geq 1 - s^*,
	\end{align*}
	where $OPT$ is now the optimal welfare, $\vec{x}$ the corresponding allocation at $OPT$, and we used the fact that $v_{\ell} \leq v_i$ for all $i \in L$.
\end{proof} 

\bigskip

Finally, \textsc{All-or-Nothing} is optimal among all truthful mechanisms for both objectives whenever the market share $s^*$ is at most $1/2$.

\begin{theorem}\label{thm:splower_revenue}
	Let $M$ be any truthful mechanism that always outputs an envy-free pricing scheme. Then the approximation ratio of $M$ for the revenue and the welfare objective is at least $2-\frac{4}{m+2}$.
\end{theorem}
\begin{proof}
	Consider an auction with equal budgets, $B$, and valuation profile $\mathbf{v}$. Assume that buyer $1$ has the highest valuation, $v_1$, buyer $2$ the second highest valuation $v_2$, with the property that $v_1 > v_2+\epsilon$, where $\epsilon$ is set later. Let $v_i < v_2$ for all buyers $i=3,4,\ldots,n$. Set $B$ such that $\lfloor \frac{B}{v_2} \rfloor=\frac{m}{2}+1$ and $\epsilon$ such that $\lfloor \frac{B}{v_2+\epsilon} \rfloor = \frac{m}{2}$. Informally, the buyers can afford $\frac{m}{2}+1$ units at prices $v_2$ and $v_2 +\epsilon$. Note that on this profile, Mechanism \textsc{All-or-Nothing} outputs price $v_2$ and allocates $\frac{m}{2}+1$ units to buyer $1$. For a concrete example of such an auction, take $m=12$, $v_1=1.12$, $v_2=1.11$ (i.e. $\epsilon = 0.01$) and $B=8$ (the example can be extended to any number of units with appropriate scaling of the parameters). 
	
	Let $M$ be any truthful mechanism, $p_M$ its price on this instance, and $p^*$ the optimal price (with respect to the objective in question). The high level idea of the proof, for both objectives, is the following. We start from the profile $\vec{v}$ above, where $p_{min}=v_2$ is the minimum envy-free price, and argue that if $p^*\neq v_2$, then the bound follows. Otherwise, $p^*=v_2$, case in which we construct a series of profiles $\vec{v},\vec{v}^{(1)},\vec{v}^{(2)},\ldots,\vec{v}^{(k)}$ that only differ from the previous profile in the sequence by the reported valuation $v_2^{(j)}$ of buyer $2$. We argue that in each such profile, either the mechanism allocates units to buyer $1$ only, case in which the bound is immediate, or buyer $2$ is semi-hungry. In the latter case, truthfulness
	and the constraints on the number of units will imply that any truthful mechanism must allocate to buyer $2$ zero items, yielding again the required bound.
	
	First, consider the social welfare objective. Observe that for the optimal price $p^*$ on profile $\mathbf{v}$, it holds that $p^*= v_2$. We have a few subcases:
	\begin{description}
		\item[\emph{Case 1}]: $p_M < v_2$. Then $M$ is not an envy-free mechanism, since in this case there would be over-demand for units. \\
		
		\item[\emph{Case 2}]: $p_M > v_2$: Then $M$ allocates units only to buyer $1$, achieving a social welfare of at most $(\frac{m}{2}+1)v_2$. The maximum social welfare is $m \cdot v_2$, so the approximation ratio of $M$ is at least $\frac{m}{(m/2)+1}=2-\frac{4}{m+2}$. \\
		
		\item[\emph{Case 3}]: $p_M = v_2$: Let $x_2$ be the number of units allocated to buyer $2$ at price $v_2$; note that since buyer $2$ is semi-hungry at $v_2$, any number of units up to $\frac{m}{2}-1$ is a valid allocation. If $x_2 =0$, then $M$ allocates units only to buyer $1$ at price $v_2$ and for the same reason as in Case 2, the ratio is greater than or equal to $2-\frac{4}{m+2}$ ; so we can assume $x_2 \geq 1$. 
		
		Next, consider valuation profile $\mathbf{v}^{(1)}$ where for each buyer $i \neq 2$, we have $v_i^{(1)}=v_i$, while for buyer $2$, $v_2<v_2^{(1)}<v_2+\epsilon$. By definition of $B$, the minimum envy-free price on $\mathbf{v}^{(1)}$ is $v_2^{(1)}$. Let $p_M^{(1)}$ be the price output by $M$ on valuation profile $\vec{v}^{(1)}$ and take a few subcases:\\
		\begin{description}%[\textbf{a}]%[\hspace{8mm}a)]'
			\setlength{\itemindent}{.01in}
			\item[a).] $p_{M}^{(1)} > v_2^{(1)}$: Then using the same argument as in Case 2, the approximation is at least $2-\frac{4}{m+2}$. \\
			\item[b).] $p_{M}^{(1)} < v_2^{(1)}$: This cannot happen because by definition of the budgets, $v_{2}^{(1)}$ is the minimum envy-free price. \\
			\item[c).] $p_{M}^{(1)} = v_2^{(1)}$: Let $x_2^{(1)}$ be the number of units allocated to buyer $2$ at profile $\mathbf{v^{(1)}}$; we claim that $x_2^{(1)} \geq 2$. Otherwise, if $x_2^{(1)} \leq 1$, then on profile $\mathbf{v}^{(1)}$ buyer $2$ would have an incentive to report $v_2$, which would move the price to $v_2$, giving buyer $2$ at least as many units (at a lower price), contradicting truthfulness. \\
		\end{description} 
		
		Consider now a valuation profile $\mathbf{v}^{(2)}$, where for each buyer $i \neq 2$, it holds that $v_i^{(2)}=v_i^{(1)}=v_i$ and for buyer $2$ it holds that $v_2^{(1)}<v_2^{(2)}<v_2+\epsilon$. For the same reasons as in Cases a-c, the behavior of $M$ must be such that: \\
		\begin{itemize}
			\setlength{\itemindent}{.2in}
			\item the price output on input $\mathbf{v}^{(2)}$ is $v_2^{(2)}$ (otherwise $M$ only allocates to buyer $1$, and the bound is immediate), and \\ 
			\item the number of units $x_2^{(2)}$ allocated to buyer $2$ is at least $3$ (otherwise truthfulness would be violated). \\
		\end{itemize} 
		
		By iterating through all the profiles in the sequence constructed in this manner, we arrive at a valuation profile $\mathbf{v}^{(k)}$ (similarly constructed), where the price is $v_2^{(k)}$ and buyer $2$ receives at least $m/2$ units. However,  buyer $1$ is still hungry at price $v_2^{(k)}$ and should receive at least $\frac{m}{2}+1$ units, which violates the unit supply constraint. This implies that in the first profile, $\vec{v}$, $M$ must allocate $0$ units to buyer $2$ (by setting the price to $v_2$ or to something higher where buyer $2$ does not want any units). This implies that the approximation ratio is at least $2-\frac{4}{m+2}$.
	\end{description}

	\noindent For the revenue objective, the argument is exactly the same, but we need to establish that at any profile $\mathbf{v}$ or $\mathbf{v^{(i)}}$, $i=1,\ldots,k$ that we construct, the optimal envy-free price is equal to the second highest reported valuation, i.e. $v_2$ or $v_2^{(i)}$, $i=1,\ldots,k$ respectively. To do that, choose $v_1$ such that $v_1 = v_2+\delta$, where $\delta > \epsilon$, but small enough such that $\lfloor \frac{B}{v_2+\delta} \rfloor = \lfloor\frac{B}{v_2}\rfloor$, i.e. any hungry buyer at price $v_2+\delta$ buys the same number of units as it would buy at price $v_2$. Furthermore, $\epsilon$ and $\delta$ can be chosen small enough such that $(\frac{m}{2}+1)(v_2+\delta) < m\cdot v_2$, i.e. the revenue obtained by selling $\frac{m}{2}+1$ units to buyer $1$ at price $v_2+\delta$ is smaller than the revenue obtained by selling $\frac{m}{2}+1$ units to buyer $1$ and $\frac{m}{2}-\epsilon$ units to buyer $2$ at price $v_2$. This establishes the optimal envy-free price is the same as before, for every profile in the sequence and all arguments go through.
	
	Given that we are working over a discrete domain, for the proof to go through, it suffices to assume that there are $m$ points of the domain between $v_1$ and $v_2$, which is easily the case if the domain is not too sparse. Specifically, for the concrete example presented at the first paragraph of the proof, assuming that the domain contains all the decimal floating point numbers with up to two decimal places suffices.
\end{proof}

\section{Impossibility Results}\label{sec:impossibility}

In this section, we state our impossibility results, which imply that truthfulness can only be guaranteed when there is some kind of wastefulness; a similar observation was made in ~\cite{borgs2005multi} for a different setting.

\begin{theorem}\label{thm:noparetoeff}
	There is no Pareto efficient, truthful mechanism that always outputs an envy-free pricing, even when the budgets are known.
\end{theorem}
\begin{proof}
	Assume by contradiction that a Pareto efficient and truthful mechanism that always outputs an envy-free price exists. Consider the following instance $I_1$ with $n=2$ and $m=3$ (the instance can be adapted to work for any number of buyers by adding many buyers with very small valuations and many items by scaling the budgets appropriately): $v_1 = v_2 = 3$ and $B_1=B_2=6$. It is not hard to see that the only Pareto efficient envy-free outcome is to set $p=3$ and allocate $2$ items to one buyer (wlog buyer 1) and $1$ item to the other buyer. Indeed, any price $p' < p$ would not be envy-free and any price $p' >p$ would sell $0$ items, yielding a utility of $0$ for both agents and the auctioneer. At the same time, any allocation that does not allocate all three items at price $p=3$ is Pareto dominated by the above allocation, since the utilities of buyers $1$ and $2$ would be $0$, but the utility of the auctioneer would be smaller.
	
	Now consider a new instance $I_2$ where $v_1=3$, $v_2=2.5$ (and it still holds that $B_1=B_2=6$). We claim that the only Pareto efficient envy-free outcome $(x,q)$ is to set the price $q=2.5$, allocate $x_1=2$ items to buyer 1 and $x_2=1$ item to buyer 2. At $(x,2.5)$, the utility of buyer $1$ is $u_1(x,2.5)=6-5 = 1$, the utility of buyer $2$ is $u_2(x,2.5)=2.5-2.5=0$ and the utility of the auctioneer is $u_a(x,2.5) = 2.5 \cdot 3 = 7.5$. 
	The only other possible allocation $x'$ at price $2.5$ would be $x'_1=2$ (since buyer 1 is hungry at price $2.5$) and $x'_2=0$, which is Pareto dominated by $(x,2.5)$. Therefore, for another Pareto efficient pair $(x',q')$ to exist, it would have to hold that $q' \neq 2.5$.
	
	Obviously, any choice $q' < 2.5$ is not envy-free and therefore we only need to consider the case when $q'>2.5$. At any such price $q'$, the utility of buyer $1$ is at most $1$, since the buyer can purchase $x'_1\leq 2$ items at a price strictly higher than $2.5$, the utility of buyer $2$ is $0$ since the price is higher than its valuation and hence it gets $x'_2=0$ items, and finally, the utility of the auctioneer is at most $6$, since it can only sell at most two items at a price no higher than $3$. This means that $(x',q')$ is Pareto dominated by $(x,2.5)$.
	
	The paragraphs above establishes that on Instance $I_1$, buyer 2 receives one item at price $3$ and on instance $I_2$, buyer 2 receives one item at price $2.5$. But then, buyer 2 would have an incentive to misreport his valuation on instance $I_1$ as being $v_2'=2.5$ and receive the same number of items at a lower price, thus increasing its utility and contradicting truthfulness.
	
	Since the proof only requires valuations and budgets to lie on points $2.5$, $3$ and $6$, the theorem also holds for the discrete domain.
\end{proof}

\bigskip

The next theorem provides a stronger impossibility result.
First, we provide the necessary definitions. A buyer $i$ on profile input $v$ is called \emph{irrelevant} if at the minimum envy-free price $p$ on $v$, the buyer can not buy even a single unit. A mechanism is called \emph{in-range} if it always outputs an envy-free price in the interval $[0,v_j]$ where $v_j$ is the highest valuation among all buyers that are not irrelevant. Finally, a mechanism is \emph{non-wasteful} if at a given price $p$, the mechanism allocates as many items as possible to the buyers.
Note that Pareto efficiency implies in-range and non-wastefulness, but not the other way around. In a sense, while Pareto efficiency also determines the price chosen by the mechanism, non-wastefulness only concerns the allocation given a price, whereas in-range only restricts prices to a ``reasonable'' interval.

\begin{theorem}\label{thm:nononwasteful}
	There is no in-range, non-wasteful and truthful mechanism that always outputs an envy-free pricing scheme, even when the budgets are known.
\end{theorem}
\begin{lemma}\label{lem:charac}
	Let $M$ be an in-range, non-wasteful and truthful mechanism. Then on any valuation profile $\mathbf{v}$ which is not trivial, $M$ must output a price $p \in \{p_{\min},p_{\min}+\gamma\}$, where $p_{\min}=\min\{p \in \mathbb{O}:p \text{ is envy-free on }\mathbf{v}\}$ and $\gamma$ is the distance between two consecutive elements of $\mathbb{O}$.
\end{lemma}

\begin{proof}
	Assume by contradiction that $M$ does not always output a price $p \in \{p_{\min},p_{\min}+\gamma\}$.  Let $\mathbf{v}=(v_1,\ldots,v_n)$ be any valuation profile that is not trivial and let $p_v$ be the price outputted by $M$; by assumption, it holds that $p_v > p_{\min}+\gamma$. By the assumption that $M$ is in-range, it holds that $v_j \geq p_v$ for some relevant buyer $j \in N$. Define 
	$$J = \{j:v_j \geq p_v : j \text{ is allocated a non-zero number of units} \}$$
	as the set of all relevant buyers with valuations at least as high as the envy-free price chosen by $M$.
	
	Now, consider an instance $\mathbf{v^{1}}$ such that $v_{i}^{1} = v_i$ for all buyers $i \in N \backslash \{j_1\}$ and $v_{j_1}^{1} = p_{\min}+\gamma$ for some buyer $j_1 \in J$, i.e. the instance obtained by $\mathbf{v}$ when some buyer $j_1 \in J$ reports a valuation equal to $p_{\min}+\gamma$. Let $p_{1}$ be the price outputted by $M$ on input $\mathbf{v^{1}}$. Note that since on instance $\mathbf{v^{1}}$ buyer $j_1$'s valuation is still higher than $p_{\min}$, it holds that $p_{\min}$ is still the minimum envy-free price in $\mathbb{O}$ on the profile $\mathbf{v^{1}}$.
	
	%First, we will argue that an envy-free price in $(0,\max_i v^1_i]$ exists on $\mathbf{v^1}$. To see this, notice that since the budgets are fixed and both $v_{j_1}$ and $v_{j_1}^{1}$ are larger than $p_{min}$, and since $p_{min}$ was an envy-free price on input $\mathbf{v}$, it is also an envy-free price on input $\mathbf{v^{1}}$, as the demand of buyer $j_1$ on profile $\mathbf{v}$ is not smaller than its demand on profile $\mathbf{v^1}$. Obviously, $0< p_{min} \leq \max_i v^1_i$.
	\begin{itemize}
		\item Assume first that $p_{1}=p_{\min}$. In that case, buyer $j_1$ on input profile $\mathbf{v}$ would have an incentive to misreport its valuation as $v_{j_1}^{1}=p_{\min}+\gamma$; that would lower the price and since $B_{j_1}$ is fixed, the buyer would receive at least the same amount of units at a lower price (since it still appears to be hungry at price $p_{\min}$). This would contradict the truthfulness of $M$.
		\item
		Now consider the case when $p_1 = v_{j_1}^1=p_{\min}+\gamma$. Note that since $p_1 > p_{\min}$, it holds that $\lfloor B_{j_1}/p_1 \rfloor \leq \lfloor B_{j_1}/p_{\min} \rfloor$, i.e. buyer $j_1$ can not demand more units at price $p_1$ compared to $p_{\min}$. On profile $\mathbf{v}$, it would be possible to allocate $\lfloor B_{j_1}/p_{\min} \rfloor$ units to buyer $j_1$ at price $p_{\min}$, therefore on profile $\mathbf{v^1}$, it is possible to allocate $\lfloor B_{j_1}/p_{1} \rfloor$ units to buyer $j_1$ at price $p_1 = p_{\min}+\gamma$. Buyer $j_1$ is semi-hungry at $p_1$ but since $M$ is non-wasteful, it must allocate at least $\lfloor B_i/p_1 \rfloor \geq \lfloor B_i/p_v \rfloor  $ units to buyer $j_1$ at a price $p_1 < p_v$, and buyer $j_1$ increases its utility by misreporting.
	\end{itemize}
	From the discussion above, it must hold that $p_1 > p_{\min}+\gamma$. For the valuation profile $\mathbf{v^1}$ (which can be seen as the different instance where buyer $1$ has deviated from $v_1$ to $p_{\min}+\gamma$), update the set $J := \{j:v_j \geq p_1: j \text{ is allocated a non-zero number of units} \}$. If $J = \emptyset$, then Mechanism $M$ is not in-range and we have obtained a contradiction. Otherwise, there must exist some other buyer $j_2 \in J$ with valuation higher than $p_1$. 
	
	Now, consider such a buyer $j_2 \in J$ and the instance $\mathbf{v^{2}}$ such that $v_{i}^{2} = v_i^{1}$ for all buyers $i \in N \backslash \{j_2\}$ and $v_{j_2}^{2}=p_{\min}+\gamma$ for buyer $j_2$, i.e. the instance obtained from $\mathbf{v^{1}}$ when some buyer $j_2$ in $J$ misreports its value being between $p_{\min}+\gamma$. Note that for the same reasons explained above, $p_{\min}$ is the minimum envy-free price in $\mathbb{O}$ on profile $\mathbf{v}^2$ as well. Let $p_2$ be the price outputted by $M$ on valuation profile $\mathbf{v^2}$. Using exactly the same arguments as we did before, we can argue that by truthfulness, it holds that $p_2 \notin \{p_{\min},p_{\min}+\gamma\}$ and therefore it must hold that $p_2 > p_{\min}+\gamma$, as every other choice is not envy-free.
	
	By iteratively considering sequences of valuations obtained in this manner, we eventually obtain an instance $\mathbf{v^{k-1}}$ such that $J = \{j_k\}$, i.e. there is only one buyer with a valuation higher than the envy-free price $p_{k-1}$ output by $M$. Repeating the argument once more will result in a valuation profile $\mathbf{v^{k}}$ where the price $p_k$ is higher than the reported valuation $v_{j_k}^{k}=p_{\min}+\gamma$ of buyer $j_k$ and the set $J$ will be empty, contradicting the fact that $M$ is in-range.  
	
	Overall, this implies that $M$ either violates truthfulness, non-wastefulness or in-range, contradicting our assumption.
\end{proof}

\bigskip

We remark here that in the continuous domain, Lemma \ref{lem:charac} can be strengthened so that $M$ can only output the minimum envy-free price, whenever it exists. Using Lemma \ref{lem:charac}, we can now prove the theorem.

\begin{proof}(of Theorem \ref{thm:nononwasteful})
	Assume by contradiction that such an in-range, non-wasteful and truthful mechanism $M$ exists. We will consider three different instances\footnote{The instances can be extended to any number of buyers by simply adding buyers with very low valuations and to many items by scaling the valuations and budgets appropriately.}  with $n=2$ and $m=3$, denoted $(v_1,v_2)$ where $v_1$ denotes the valuation of buyer 1 and $v_2$ denotes the valuation of buyer 2, with budgets $B_1=B_2=6+2\gamma$. 
	
	First, consider the instance $(2.5,2.5)$ and note that since the instance is not trivial and since the minimum envy-free price is $2.5$, by Lemma \ref{lem:charac}, the price chosen by $M$ for this instance must be either $2.5$ or $2.5 + \gamma$. Furthermore, since $M$ is in-range, the price can not be $2.5 + \gamma$, therefore the price chosen on $(2.5,2.5)$ is $2.5$. Since $M$ is non-wasteful and each buyer can afford exactly $2$ items at price $2.5$ and there are $3$ available items, one buyer (wlog buyer 1) gets allocated $2$ items and the other buyer (wlog buyer 2) gets allocated $1$ item at this price.
	
	Now consider the instance $(3,2.5)$ and note that since it is not trivial and since again, $2.5$ is the minimum envy-free price, $M$ must either output $2.5$ or $2.5 + \gamma$ as the price. Assume first that $M$ selects the price to be $2.5+\gamma$. Since buyer $1$ is hungry at this price and can afford to buy exactly $2$ units, its allocation on instance $(3,2.5)$ is $2$ units at price $2.5+\gamma$. But then, on instance $(3,2.5)$ buyer 1 would have an incentive to misreport its valuation as being $2.5$ since on the resulting instance, which is $(2.5,2.5)$, it still receives $2$ items at a lower price, increasing its utility. Note that if it was buyer 2 that received $2$ items on instance $(2.5,2.5)$, we could have made the same argument using instance $(2.5,3)$ instead.
	
	Finally, assume that on instance $(3,2.5)$, $M$ outputs $2.5$ as the price. By non-wastefulness, buyer 2 receives exactly $1$ unit at this price. But then, consider the instance $(3,3)$, where, using the same arguments as in the case of instance $(2.5,2.5)$, Mechanism $M$ must output $3$ as the price and allocate $2$ units to one buyer and $1$ unit to the other buyer. Crucially, both buyers have utility $0$ on instance $(3,3)$. But then, buyer 2 could misreport its valuation as being $2.5$, resulting in instance $(3,2.5)$ where it receives $1$ unit at a price lower than its actual valuation, benefiting from the misreport. This contradicts truthfulness.
	
	Assume by contradiction that such an in-range, non-wasteful and truthful mechanism $M$ exists. Consider the same instance $I_1$ as the one used in the proof of Theorem \ref{thm:noparetoeff}, with $n=2$, $m=3$ and $v_1=v_2=3$ and $B_1=B_2=6+2\gamma$. (Again the proof can be generalized to many agents and units similarly to the proof of Theorem \ref{thm:noparetoeff}). By Lemma \ref{lem:charac} and since $I_1$ is not trivial, $M$ must either output $p=3$ or $p=3+\gamma$ and by the fact that it is in-range, it must output $p=3$. Since $M$ is non-wasteful, it must allocate $2$ units to one of the buyers with valuation $3$ (wlog buyer 1) and $1$ unit to the other buyer.
	
	Now consider an instance $I_{2a}$ where $v'_1=3$ and $v'_2=2.5$. Since $2.5$ is now the minimum envy-free price and $I_2$ is again not trivial, $M$ must output either $p'= 2.5$ or $p=2.5 + \gamma$. We will obtain a contradiction for each case. Assume first that $p'=2.5$; since buyer 1 is hungry, it must hold that $x'_1=2$ and by non-wastefulness, it must hold that $x'_1=1$. In that case however, for the same reason explained in the proof of Theorem \ref{thm:noparetoeff}, $v'_2=2.5$ could be a beneficial deviation of buyer 2 on instance $I_1$, violating truthfulness.
	Now we argue for the case when $p'=2.5+\gamma$. Consider the instance $I_3$ where $\bar{v}_1=\bar{v}_2=2.5$. Since $M$ is in-range and $I_3$ is not trivial, $M$ must select price $\bar{p}=2.5$, since every other price is either not envy-free, or higher than all the valuations. By non-wastefulness, one buyer must receive 2 units at $\bar{p}$ and the other agent must receive $1$ unit (because each buyer can afford exactly 2 units and there are 3 units available). If buyer $1$ receives $2$ units, i.e. $\bar{x}_1=2$, misreporting its valuation on instance $I_{2a}$ as $2.5$ would give the buyer higher utility, since it gets allocated the same number of items at a lower price.
	It remains to deal with the case when on instance $I_3$, buyer 1 is allocated $1$ item and buyer 2 is allocated $2$ items, i.e. $\bar{x}_1=1$ and $\bar{x}_2=2$. 
	
	Now consider the instance $I_{2b}$ where $\hat{v}_1=2.5$ and $\hat{v}_2=3$, i.e. instance $I_{2b}$ is exactly the same as instance $I_{2a}$ with the indices of the two buyers swapped. Again, since instance $I_{2b}$ is not trivial, by Lemma \ref{lem:charac}, $M$ must output a price $\hat{p} \in \{2.5,2.5+\gamma\}$. If $\hat{p} = 2.5+\gamma$, then we consider again Instance $I_3$. Since on that instance $\bar{p}=2.5$ and $\bar{x}_2=2$ by the assumption above, buyer 2 has an incentive to misreport its valuation on instance $I_{2b}$ as being $2.5$, contradicting truthfulness. Therefore, it must hold that $\hat{p}=2.5$ on instance $I_{2b}$.
	
	However, by non-wastefulness, buyer $1$ receives one unit at price $\hat{p}$ on instance $I_{2b}$, i.e. $\hat{x}_1 =1$. We will consider the $2.5$ as a potential deviation of buyer 1 on instance $I_1$ (where its true valuation is $v_1=3$). The utility of the buyer before misreporting is $0$ (since the chosen price on instance $I_1$ is $p=3$) whereas the utility after misreporting is $3-2.5 = 0.5$, i.e. strictly positive. Therefore, buyer $1$ has a beneficial deviation on instance $I_1$, violating the truthfulness of $M$.
	
	By truthfulness, it must also hold that $\bar{p} \geq 2.5+\gamma$, otherwise on instance $I_2$ buyer 1 would have an incentive to misreport its valuation as $2.5+\gamma$ and still receive $2$ items at a lower price (since at any price $p<2.5+\gamma$ buyer 1 on instance $I_3$ is hungry). From the discussion above, it must hold that $\bar{p}=2.5+\gamma$ and by non-wastefulness and since buyer 1 can afford two items at price $2.5+\gamma$, it must hold that $\bar{x}_1=2$.
\end{proof}

%To prove the impossibility, we first obtain a necessary condition; any mechanism in this class must essentially output the minimum envy-free price (or the next highest price on the output grid). Then we can use this result to construct and example where the mechanism must leave some items unallocated in order to satisfy truthfulness.

\section{Discussion}

Our results show that it is possible to achieve good approximate truthful mechanisms, under reasonable assumptions on the competitiveness of the auctions which retain some of the attractive properties of the Walrasian equilibrium solutions. The same agenda could be applied to more general auctions, beyond the case of linear valuations or even beyond multi-unit auctions. It would be interesting to obtain a complete characterization of truthfulness in the case of private or known budgets; for the case of private budgets, we can show that a class of order statistic mechanisms are truthful, but the welfare or revenue guarantees for this case may be poor. In the appendix we present an interesting special case, that of \emph{monotone auctions}, in which Mechanism \textsc{All-Or-Nothing} is optimal among all truthful mechanisms for both objectives, regardless of the market share.

\section{Acknowledgements}

We would like to thank the MFCS reviewers for useful feedback.

\addcontentsline{toc}{section}{\protect\numberline{}References}%

\bibliographystyle{alpha}

\bibliography{multiunitarxiv}

% Appendix
\appendix

\section*{APPENDIX}

\section{Monotone Auctions}\label{sec:interesting}

In the main text, we proved the approximation ratio guarantees of Mechanism {\sc All-or-Nothing}, as a function of the market share. In this section, we will examine the case of common budgets and the more general class of monotone auctionss:
\begin{itemize}
	\item The budgets are \emph{common} when $B_i=B$ for all buyers $i \in N$.
	\item The budgets are monotone in the valuations when  $v_i \geq v_j \Leftrightarrow B_i \geq B_j$. We call such auctions \emph{monotone}.
\end{itemize} 
Note that the second case is more general than the first, where for the right-hand side we have $B_i = B_j$ for all $i,j \in N$. We will prove that for those cases, Mechanism {\sc All-or-Nothing} is optimal among all truthful mechanisms, for both the welfare and the revenue objective. For the welfare objective, the approximation ratio guarantee will be completely independent of the market share. For the revenue objective, the dependence will be rather weak; we prove that the bound holds in all auctions except monopsonies. A monopsony is an auctionin which a single buyer can afford to buy all the items at a very high price.

\begin{definition}
	An auction is a \emph{monopsony}, if the buyer with the highest valuation $v_1$ has enough budget $B_1$ to buy all the units at a price equal to the second highest valuation $v_2$. 
\end{definition}

Note that when the market is not a monopsony, that implies that the market share $s^*$ is less than $1$.\footnote{Note that instead of ruling out monopsonies, another approach would be to consider a different benchmark, that does not include the case of an omnipotent buyer, like the EFO$^{(2)}$ benchmark for revenue, see \cite{hartline2013mechanism}, Chapter 6.}

\begin{theorem}\label{thm:special cases}
	The approximation ratio of Mechanism \emph{All-or-Nothing} for monotone auctions is 
	\begin{itemize}
		\item at most $2$ for the social welfare objective.
		\item at most $2$ for the revenue objective when the auction is not a monopsony.
	\end{itemize} Furthermore, no truthful mechanism can achieve an approximation ratio smaller than $2 - \frac{4}{m+2}$ even in the case of common budgets. 
\end{theorem}

\begin{proof}
	First, note that the profile constructed in Theorem \ref{thm:splower_revenue} is one where the budgets are common and therefore the lower bound extends to both cases mentioned above. Therefore, it suffices to prove the approximation ratio of Mechanism {\sc All-or-Nothing} for both objectives, when the auction is monotone.
	
	We start from the social welfare objective and consider an arbitrary profile $\vec{v}$. Without loss of generality, we can assume that $\vec{v}$ is not trivial (otherwise the optimal allocation allocates $0$ items in total) and note that the optimal envy-free price is $p^*=p_{min}$ and let $\vec{x}$ be the corresponding optimal allocation. Following the arguments in the proof of Theorem \ref{thm:revenuebound}, we establish that the according to $\vec{x}$ at most one additional semi-hungry buyer is allocated a positive number of units, compared to the allocation of Mechanism {\sc All-or-Nothing}; let $\ell$ be that buyer and let $x_{\ell}$ be its optimal allocation. %Furthermore, let $n_1$ be the number of essentially hungry buyers under the optimal allocation at price $p_{min}$. 
	
	The social welfare loss of Mechanism {\sc All-or-Nothing} is $x_\ell \cdot v_{\ell} \leq v_\ell \cdot \lfloor B_\ell/v_{\ell} \rfloor$, i.e. the contribution of the the semi-hungry buyer that receives $0$ items by {\sc All-or-Nothing}, in contrast to the optimal allocation. Since the profile $\vec{v}$ is not trivial, there exists at least on other buyer $j$ that receives $\min\{m, \lfloor B_j/v_{\ell}\rfloor\}$ units in the optimal allocation $\vec{x}$. If it receives $m$ units, then $x_\ell = 0$ and the ratio on the profile is $1$. Otherwise, the contribution to the welfare (for both the optimal allocation and the allocation of {\sc All-or-Nothing}) from buyer $j$ is $v_j \cdot \lfloor B_j/v_{\ell}\rfloor\} \geq v_j \cdot \lfloor B_\ell/v_{\ell}\rfloor\}$, since $v_\ell \leq v_j \Leftrightarrow B_\ell \leq B_j$ by the monotonicity of the auction. Then we have:
	
	\begin{eqnarray*}
		\frac{\mathcal{SW}(AON)}{OPT} &\geq &\frac{OPT- v_\ell \cdot \lfloor B_\ell/v_{\ell} \rfloor}{OPT}= 1 - \frac{v_\ell \cdot \lfloor B_\ell/v_{\ell} \rfloor}{OPT} \\
		& \geq & 1 - \frac{v_\ell \cdot \lfloor B_\ell/v_{\ell} \rfloor}{(v_\ell+v_j) \cdot \lfloor B_\ell/v_{\ell} \rfloor} = 1 - \frac{v_\ell}{v_j+v_\ell} \geq \frac{1}{2}.
	\end{eqnarray*}
	
	For the revenue objective, again let $p^*$ be the optimal envy-free price and let $\vec{x}$ be the corresponding allocation. We consider two cases:
	\begin{itemize}
		\item $p^*= p_{min}$: The argument in this case is very similar to the one used above for the social welfare objective. In particular,  since $p^*=p_{min}=v_\ell$, we now have that the loss in revenue from the semi-hungry buyer $\ell$ for Mechanism {\sc All-or-Nothing} is at most $x_\ell \cdot v_{\ell} \leq v_\ell \cdot \lfloor B_\ell/v_{\ell} \rfloor$ whereas the contribution from buyer $j$ is $v_{\ell} \cdot \lfloor B_j/v_{\ell}\rfloor\}$, which is at most $v_{\ell} \cdot \lfloor B_\ell/v_{\ell}\rfloor\}$ by the monotonicity of the auction. Therefore, we have that:
		
		\begin{eqnarray*}
			\frac{\mathcal{REV}(AON)}{OPT} &\geq& \frac{OPT- v_\ell \cdot \lfloor B_\ell/v_{\ell} \rfloor}{OPT}= 1 - \frac{v_\ell \cdot \lfloor B_\ell/v_{\ell} \rfloor}{OPT} \\
			& \geq & 1 - \frac{v_\ell \cdot \lfloor B_\ell/v_{\ell} \rfloor}{2v_\ell \cdot \lfloor B_\ell/v_{\ell} \rfloor} = 1 - \frac{v_\ell}{2v_\ell} =\frac{1}{2}.
		\end{eqnarray*}
		
		\item $p^* > p_{min}$. In that case, the argument is exactly the same as in Case 2 of the proof of Theorem \ref{thm:revenuebound}, which holds when the market share is less than $1$, i.e. when the auction is not a monopsony. 
	\end{itemize}
\end{proof}

\noindent To complete the picture, we prove in the following that if the auction is a monopsony, the approximation ratio of any truthful mechanism is unbounded. This can be captured by the following theorem.

\begin{theorem}\label{thm:badratio}
	If the auction is a monopsony, the approximation ratio of any truthful mechanism for the revenue objective is at least $\mathcal{B}$ for any $\mathcal{B}>1$, even if the budgets are public. %This holds for both the discrete and the continuous version.
\end{theorem}

\begin{proof}
	Consider the following monopsony. Let $i_{1}=\argmax_{i} v_i$, for $i=1,\ldots,n$ be a single buyer with the highest valuation and denote $v_{i_1}=v_1$ for ease of notation. Similarly, let $i_2 \in \argmax_{i \in N\backslash\{i_1\}} v_{i}$ be one buyer with the second largest valuation and let $v_{i_2}=v_2$. Furthermore, let $v_1 > v_i$ for all $i \neq i_1$ and $B_{i_1} = p \cdot m$, for some $v_2<p\leq v_{1}$  i.e. buyer $i_1$ can afford to buy all the units at some price $p > v_2$. Additionally, let $B_{i_2} \geq v_2$, i.e. buyer $i_2$ can afford to buy at least one unit at price $v_2$.\footnote{Note that setting $B_{i_2}=B_{i_1}$ satisfies this constraint and creates an auction with identical budgets, so the proof goes through for that case as well.} Finally, for a given ${\mathcal B}>1$ let $v_2$ and $p$ be such that $\mathcal{B}=p/v_2$. Note that the revenue-maximizing envy-free price for the instance $\mathbf{v}$ is at least $p$ and the maximum revenue is at least $p \cdot m$.
	
	Assume for contradiction that there exists a truthful mechanism $M$ with approximation ratio smaller than $\mathcal{B}$ and let $p^*$ be the envy-free price output by $M$ on $\mathbf{v}$. Since $p^*$ is envy-free and $B_{i_1} > v_{1}\cdot m$ and $B_{i_2}\geq v_2$, it can not be the case that $p^* < v$, otherwise there would be over-demand for the units. Furthermore, by assumption it can not be the case that $p^*=v_2$ as otherwise the ratio would be $\mathcal{B}$ and therefore it must hold that $p^*> v_2$. 
	
	Now let $\mathbf{v'}$ be the instance where all buyers have the same valuation as in $\mathbf{v}$ except for buyer $i_1$ that has value $v_{1}'$ such that $v < v_{1}' < p^*$ and let $\tilde{p}$ be the envy-free price that $M$ outputs on input $\mathbf{v'}$. If $\tilde{p} > v_{1}'$, then the ratio of $M$ on the instance $\mathbf{v'}$ is infinite, a contradiction. If $\tilde{p} \leq v_1'$ and since $\tilde{p}$ is envy-free, it holds that $v_2\leq \tilde{p} < p^*$. In that case however, on instance $\mathbf{v}$, buyer $i_1$ would have an incentive to misreport its valuation as $v_{1}'$ and reduce the price. The buyer still receives all the units at a lower price and hence its utility increases as a result of the devation, contradicting the truthfulness of $M$. 
\end{proof}

\section{Computational Results}

In this section we study the problem of computing a welfare and revenue maximizing envy-free pricing without incentives. 
A welfare maximizing envy-free pricing can be computed efficiently using the following algorithm. 

\bigskip

\noindent \emph{\textbf{Algorithm 2:} Find the minimum envy-free price in the set $\mathcal{P}$ and allocate maximally to the buyers (using greedy allocation and lexicographic tie-breaking for the indifferent buyers).}

\begin{theorem}
	Algorithm 2 computes a welfare maximizing envy-free pricing in time polynomial in $n$ and $\log m$ for linear multi-unit markets.
\end{theorem}
The correctness and running time of the algorithm are based on the observation that the social welfare is non-increasing in the price which together with Lemma \ref{lem:small_to_large} allows us to search for the minimum envy-free price in the set $\emph{P}$. This can be done by first checking all $n$ values $v_i$, for $i=1,\ldots,n$ and then for each buyer $i$, by checking over the ordered set $\left\{\floor{\frac{B_i}{m}},\floor{\frac{B_i}{m-1}},\ldots,\floor{B_i}\right\}$ using binary search.

\medskip

For revenue we design a fully polynomial time approximation scheme as well as an exact algorithm that runs in polynomial time for constantly many types of buyers. Without loss of generality, the buyers can be assumed to be ordered by their valuations: $v_1 \geq v_2 \geq \ldots \geq v_n$.

\newpage
\subsection{Fully Polynomial Time Approximation Scheme for Revenue} \label{app:fptas}

\begin{theorem} \label{thm:max_revenue}
	There exists an FPTAS for computing a revenue-maximizing envy-free pricing in linear multi-unit markets.
\end{theorem}
The main idea is to divide the analysis in two cases, depending on whether the number of units is large (at least $n/\epsilon$) or small. When the number of units is large, we solve a continuous variant of the problem, where the good is viewed as a continuous resource valued uniformly by each buyer $i$, at total value $v_i \cdot m$. The buyers have the same budgets $B_i$. The continuous problem can be solved efficiently, and by rounding one obtains a solution in the discrete instance with bounded loss. When the number of units is small, we can iterate over the set $\mathcal{P}$ of candidate prices in polynomial time and find a revenue-maximizing envy-free price and allocation.

Consider the \emph{continuous linear multi-unit market}, defined above. Note that in such a market, every interested buyer must receive exactly a $(B_i/p)$-fraction of the good, since the allocation does not have to be integer. As shown next, the revenue-maximizing envy-free price in the continuous variant can be computed in polynomial time.

\begin{lemma} \label{lem:price_opt_val}
	A revenue maximizing envy-free price for the continuous linear multi-unit market can be computed in polynomial time. Moreover, the price is equal to a valuation. 
\end{lemma}
%\vspace{-4mm}
\begin{proof}
	Recall the valuations are ordered: $v_1 \geq v_2 \geq \ldots \geq v_n$. Let $p$ be a revenue maxinimizing envy-free price for the continuous problem. If $p = v_i$ for some $i \in N$ the lemma trivially follows. Otherwise, $p \in (v_{\ell},v_{\ell+1})	$ for some $\ell \in N$. Since the resource is continuous and $p$ maximizes revenue, the budgets of buyers $1,2,\ldots, \ell$ are exhausted at price $p$, while buyers $\ell_1,\ldots,n$ have zero demand and the revenue is $\sum_{i=1}^{\ell} B_i$. By setting the price to $v_{\ell}$, we can obtain exactly the same revenue, since the set of buyers that purchases any units is still $\{1,2,\ldots,\ell\}$ and the budgets of all buyers in this set are still exhausted. By Lemma \ref{lem:small_to_large}, $v_{\ell}$ is an envy-free price and by the argument above, it maximizes revenue.
\end{proof} 

\noindent The next lemmas concern the revenue achieved in the continuous problem and enable connecting the maximum revenue in the continuous problem to the discrete one.

\begin{lemma} \label{lem:cont_opt}
	The optimal revenue in a continuous linear multi-unit market is at least as high as in the discrete version.
\end{lemma}
%\vspace{-4mm}
\begin{proof}
	Let $(p, x)$ be a price and allocation that maximize revenue for the discrete problem; denote by $R$ the revenue obtained.
	Set the price $p$ in the continuous version and consider two cases:

	\emph{Case 1}: The demand is (weakly) lower than the supply at $p$. Then the same revenue $R$ can be obtained in the continuous version. 	
	
	\emph{Case 2}: The demand is higher than the supply. Then we can continuously increase the price until the demand can be met as follows. Let $p'$ be the current price in the continuous problem. If $p' \in (v_{\ell}, v_{\ell+1})$, then by continuously increasing $p'$ the demand 
	decreases continuously. If the demand is never met in the interior of this interval, then when the price reaches $v_{\ell}$, we can continuously
	decrease the allocation of buyer $\ell$ until either reaching an envy-free pricing or making buyer $\ell$'s allocation zero (the case where there
	are multiple buyers with valuation $v_{\ell}$ is handled similarly, by decreasing their allocations in some sequence). In the latter case, 
	since the demand still exceeds the supply, we iterate by increasing the price continuously in the interval $[v_{\ell - 1}, v_{\ell})$. 
	
	From the two cases, the optimal revenue in the discrete problem is no higher than in the continuous problem, which completes the proof.
\end{proof}

\begin{lemma} \label{lem:opt_rev_exhausted}
	Given the optimal revenue $R$ for the continuous linear multi-unit market, setting the price to $p = R/m$ gives an envy-free pricing scheme with the same revenue where the whole resource is sold.
	%The optimal revenue in the continuous problem can be supported at a price and allocation where the entire cake is sold. Moreover, such a price and allocation
	%can be computed in polynomial time.
\end{lemma}
%\vspace{-4mm}
\begin{proof}
	By Lemma~\ref{lem:price_opt_val}, the optimal revenue in the continuous problem can be computed in polynomial time by inspecting all the valuations.
	Let $R$ denote the optimal continuous revenue. If none of the prices where the optimal revenue can be obtained support an allocation at which the entire resource is sold, then the price can be decreased
	continuously (skipping the valuation points as they have been considered before) while maintaining the revenue constant until the entire resource is sold.
	By continuity (taking the valuation points into account), there exists a revenue maximizing price at which the whole resource is sold.
	
	This problem can be solved in polynomial time by finding the buyer $\ell \in [n]$ with the property that $\sum_{i=1}^{\ell} B_i = R$. Set $p = \frac{R}{m}$ and 
	compute the corresponding allocation, where buyers $\ell + 1, \ldots, n$ don't receive anything. 
\end{proof}

\medskip

\noindent The FPTAS for revenue is given as Algorithm \ref{alg:FPTAS} and the proof of correctness is as follows.

%\vspace{6mm}
\begin{algorithm}[h!]
	\begin{algorithmic}[1]
		\caption{\textsc{ FPTAS Revenue-Maximizing-EF-Pricing($N, \vec{v}, \vec{b}, m$)}}
		\label{alg:FPTAS}
		\Require Buyers $N$ with linear valuations $\vec{v}$, budgets $\vec{b}$, $m$ units and $\epsilon > 0$.
		\Ensure An envy-free pricing $(p^*, \vec{x}^*)$ with revenue at least $(1-\epsilon)\mathcal{R}$, where $\mathcal{R}$ is the maximum revenue. %obtained at any envy-free price.
		%\SetVline
		\If {$m \leq \frac{n}{\epsilon}$} 
		\For{each price $p\in \mathcal{P}$}
		\State Check if $p$ is envy-free. If not, continue with the next candidate price.
		\State $\vec{x}_p \leftarrow \textsc{Compute-EF-Allocation}(p)$ %\Comment{Given a price, an EF allocation can be computed in $O(n)$} 
		\State $\mathcal{R}_p \leftarrow$ Revenue of $(\vec{x}_p,p)$.
		\EndFor 
		\State $p^* \leftarrow \argmax_{p} \mathcal{R}_p$.
		\State \textbf{return} $(p^*, \vec{x}_{p^*})$
		\Else
		\State $(\vec{x},p) \leftarrow$ {\sc Compute-Continuous}$(N, \vec{v}, \vec{b}, m)$ 
		\State For all $\in N$, let $\bar{x}_i = \lfloor x_i \rfloor$. \Comment Round the allocation down to integers
		\State \textbf{return} $(p, \vec{\bar{x}})$.
		\EndIf 
	\end{algorithmic}
\end{algorithm}

%		\captionsetup[algorithm]{labelformat=empty}

\begin{algorithm}[h!]
	\begin{algorithmic}[1]
		\Algphaze{\textsc{Compute-Continuous}$(N, \vec{v}, \vec{b}, m)$  \Comment{\textnormal{function that computes a revenue-maximizing price where everything is sold in the continuous case}}}
		\For {$i=1$ to $n$}
		\State Check if the price $v_i$ is envy-free. If not, continue with the next valuation.
		\State $\vec{x}^i \leftarrow \textsc{Compute-EF-Allocation}(v_i)$.
		\State $\mathcal{R}^i \leftarrow$ Revenue of $(\vec{x}^i,v_i)$.  
		\EndFor
		\State $\mathcal{R} \leftarrow \max_{i}\mathcal{R}^i$
		\State $p^* \leftarrow \mathcal{R}/m$.
		\State $\vec{x}^* \leftarrow \textsc{Compute-EF-Allocation}(p^*)$.
		
		\State \textbf{return}$(p^*,\vec{x}^*)$ 
		%\textbf{define} \textsc{Search}$(a,b, \Delta)$ \emph{// call to a recursive subroutine}\\
		
		\Algphase{\textsc{Compute-EF-Allocation}$(p)$ \Comment{\textnormal{function that computes an envy-free allocation at price $p$}}}
		\For{each buyer $i \in N$ such that $v_i >p$} \Comment{Interested buyers}
		\State $x_i \leftarrow D_i$ \Comment{Demand sets are singletons}
		\EndFor
		\For{each buyer $i \in N$ such that $v_i=p$} \Comment{Indifferent buyers}
		\If {$\text{\#available units} > 0$}
		\State $x_i \leftarrow \min\{\text{\#available units},D_i\}$
		\EndIf
		\EndFor
	\end{algorithmic}
\end{algorithm} 

%\vspace{10mm}

\begin{proof}(of Theorem \ref{thm:max_revenue})
	We show that Algorithm \ref{alg:FPTAS} is such an FPTAS.
	Let $\epsilon > 0$ and consider two cases, depending on whether the number of units is \emph{small} or \emph{large}.
	
	\emph{Case 1}: $m \leq \frac{n}{\epsilon}$.
	In this case we can just iterate over all the possible prices in $\mathcal{P}$ and select  
	the revenue maximizing envy-free price and the corresponding allocation. This step can be done in time $O\left(\frac{n^2}{\epsilon}\right)$.
	
	\emph{Case 2}: $m > \frac{n}{\epsilon}$. In this case we first solve the continuous linear multi-unit market optimally by finding the price and allocation $(p, \vec{x})$ where the whole resource is sold.
	This can be done in polynomial time by Lemma~\ref{lem:opt_rev_exhausted}.
	%The revenue in the continuous problem is
	Now consider the allocation $\tilde{\vec{x}}$, given by $\tilde{x}_i = \lfloor x_i \rfloor$, i.e. the allocation obtained from $\vec{x}$ if we round all fractions down to the nearest integer value.
	Note that $(p, \tilde{\vec{x}})$ is an envy-free pricing in the discrete problem, since 
	each buyer receives their demand. We argue that $(p, \tilde{\vec{x}})$ approximates the optimal revenue of the discrete instance within a factor of $1 - \epsilon$.
	
	%By Lemma~\ref{lem:cont_opt}, it is sufficient to show that $(p, \tilde{\vec{x}})$ approximates the optimal continuous revenue within the same factor.
	
	Let $OPT^c$ and $OPT^d$ denote the optimal revenues in the continuous and discrete instances, respectively, and $R^d$ the revenue obtained by the rounding procedure above. 
	
	%	\newpage
	
	We obtain the following inequalities:
	\begin{eqnarray*}
		\frac{R^d}{OPT^d} &\geq& \frac{OPT^c - p \cdot n}{OPT^d} 
		= \frac{p \cdot m - p \cdot n}{OPT^d}  
		\geq \frac{p \cdot m - p \cdot n}{OPT^c}\\
		&\geq& \frac{p \cdot m - p \cdot n}{p \cdot m} 
		= 1 - \frac{n}{m} 
		> 1 - \epsilon
	\end{eqnarray*}
	\noindent The first inequality holds because the rounding procedure only loses at most $pn$ revenue,
	the second identity holds because all the units are sold at price $p$, and
	the third inequality holds by Lemma~\ref{lem:cont_opt}.
	This step can be completed in time $O(n^2)$. The total runtime is bounded by the maximum in each case, which is $O\left(n^2/\epsilon\right)$.
	%By the two cases, we have that the total runtime is bound by (TBD).
\end{proof}

\newpage

\subsection{Exact Algorithm for Revenue}\label{app:revenue}
In this section we provide an exact algorithm for computing a revenue maximizing (Walrasian) envy-free pricing. This algorithm runs in polynomial time when the number of types of buyers is fixed.

\begin{theorem} \label{thm:max_revenue_constant}
	Given a linear multi-unit market, a revenue-maximizing envy-free pricing can be computed in polynomial time
	when the number of types (of buyers) is fixed.
\end{theorem}

The following problem generalizes the problem of computing a revenue-maximizing envy-free pricing scheme.

\setcounter{problem}{0}

\begin{problem} \label{problem1}
	Given $\alpha_1,\alpha_2,...,\alpha_\ell \in \mathbb{Q}$, output $x \in \arg\max \sum_{i=1}^\ell \frac{\floor{\alpha_ix} }{x}$, $x \in [a,b] \cap \mathbb{N}$.
\end{problem}

\begin{lemma}\label{lem:aproblem}
	If Problem \ref{problem1} can be solved in polynomial time, then the problem of finding a revenue-maximizing envy-free pricing scheme can be solved in polynomial time as well. 
\end{lemma}

\begin{proof}
	First, recall that by Lemma \ref{lem:inthesetP}, the revenue-maximizing envy-free price $p$ is in set $\mathcal{P}$, i.e. it is either equal to some valuation $v_i$ or equal to some fraction $B_i/k$, for some $i \in N$ and some $k \in [1,m] \cap \mathbb{N}$, which means that at least one of the budgets is exhausted at that price and allocation. In the former case, we can simply iterate over all $n$ values $v_1, \ldots, v_n$ and the revenue-maximizing price can be found in polynomial time and by Lemma \ref{lem:frompricetoalloc}, so can the corresponding revenue-maximizing allocation and the lemma holds. Therefore, we will assume that $p=B_i/k$ for some $i \in N$ and some integer $k\leq m$. We will show how to find the revenue-maximizing price given that some buyer's $j$ budget is exhausted; then we can iterate over all the buyers and calculate the revenue obtained in each case, outputting the price and allocation with the maximum revenue.
	
	Let $\alpha_i = B_i/B_j$ and let $x$ be the allocation of buyer $j$. From the discussion above, it holds that $p=B_j/x$. Assuming that buyer $i$ is hungry at price $p$, the revenue obtained by buyer $i$ is 
	\[
	p \cdot \left \lfloor\frac{B_i}{p} \right \rfloor= \frac{B_j}{x} \cdot \left \lfloor \frac{\alpha_i \cdot B_j}{x} \right \rfloor = B_j \frac{\left\lfloor \alpha_i x\right \rfloor}{x},
	\]
	The total revenue is $\mathcal{REV} =B_j \cdot  \left(\sum_{i=1}^{l}\frac{\left\lfloor \alpha_i x\right \rfloor}{x}\right)$, where $l$ is the number of hungry buyers at price $p=B_j/x$.
	
	Now, given that a revenue-maximizing envy-free price is not equal to some valuation $v_i$, it lies in an interval $(v_{i+1},v_{i})$ for some $i \in N$. For each such interval $\mathcal{I}$, the set of hungry buyers $S_\mathcal{I}$ at any chosen price consists of the buyers with valuations $v_j \geq v_{i}$ and hence we know exactly the members of this set. By the discussion above, by letting $\ell = |S_\mathcal{I}|$ in Problem \ref{problem1}, the value of the quantity $\sum_{i=1}^\ell \frac{\floor{a_ix} }{x}$ gives exactly the maximum revenue attained in the interval $(v_{i+1},v_{i})$. By iterating over all intervals, we can find the revenue-maximizing price and allocation.
	
	What is left to show is that the envy-freeness constraint can be captured by the constraint $x \in [a,b] \cap \mathbb{N}$. Note that the fact that the price lies in some open interval between two valuations and the envy-freeness constraint define an interval $[c,d]$ from which the price has to be chosen. The value of $c$ is either the minimum envy-free price in $\mathcal{P}$ or the smallest price in $\mathcal{P}$ strictly larger than $v_{i+1}$ whereas the value of $d$ is the largest price in $\mathcal{P}$ strictly smaller than $v_i$. All of these prices can be found in polynomial time by binary search on the set $\mathcal{P}$, similarly to the algorithm for maximizing welfare. Therefore, the price constraint can be written as $p \in [c,d] \cap \mathcal{P}$ and by setting $a=B_j/d$ and $b=B_j/c$, we obtain the corresponding interval of Problem \ref{problem1}. 
\end{proof}

\bigskip

By Lemma \ref{lem:aproblem}, it is sufficient to prove that Problem \ref{problem1} can be solved in polynomial time; the same algorithm will also find a revenue-maximizing envy-free pricing scheme. By setting $k_i = \floor{a_ix}$, from Problem \ref{problem1} we obtain the following equivalent problem:

\begin{problem}\label{problem2}
	Given $a,b,\alpha_i \in \mathbb{Q}$,
	\begin{eqnarray*}
		\max && \frac{k_1+...+k_\ell}{x} \\
		s.t. && k_i \leq \alpha_ix < k_i+1, ~~i=1,2,...,\ell\\
		&& a \leq x \leq b \\
		&& x,k_i \in \mathbb{Z},~~i=1,2,...,\ell    
	\end{eqnarray*}
\end{problem}

\noindent Problem \ref{problem2} is obviously an optimization problem with a non-linear objective function. We will consider the decision version of the problem, where the value of the objective function is restricted to lie in some interval.

\begin{problem}\label{problem3}
	Given $a,b,c,d,\alpha_i \in \mathbb{Q}$, decide if there exist $x, k_1, \ldots, k_n \in \mathbb{Z}$ such that 
	\begin{eqnarray*}
		&& cx \leq k_1+...+k_n \leq dx\\
		&& k_i \leq \alpha_ix < k_i+1, ~~i=1,2,...,n\\
		&& a \leq x \leq b \\
		&& x,k_i \in \mathbb{Z},~~i=1,2,...,n
	\end{eqnarray*}
\end{problem}

We are now ready to prove Theorem \ref{thm:max_revenue_constant}.

\begin{proof}(of Theorem \ref{thm:max_revenue_constant})
	We will prove that Problem \ref{problem2} can be solved in polynomial time, by proving that Problem \ref{problem3} can be written as an Integer Linear Program (ILP) problem and then providing an algorithm for Problem \ref{problem2} given an algorithm for Problem \ref{problem3}. Since Problem \ref{problem1} and Problem \ref{problem2} are equivalent (by Lemma \ref{lem:aproblem}), this will also establish the existence of a polynomial time algorithm for finding a revenue-maximizing envy-free pricing scheme. First, the ILP problem is known to be solvable in polynomial time when the number of variables is fixed.
	
	\begin{lemma}[\cite{lenstra1983integer}]
		The ILP problem with a fixed number of variables can be solved in polynomial time.
	\end{lemma}
	
	In order to show that Problem \ref{problem3} can be written as an ILP, we need to handle the non-integer coefficients. For every coefficient $j \in \mathbb{Q}$, we multiply their common denominator in each equation and all coefficients become integers, without affecting the length of the input. The inequalities that have the ``$\geq$'' direction can be handled easily. To handle the strict inequality constraint, for each such constraint such that $\mathcal{A}<\mathcal{B}$, we instead write $2\mathcal{A} \leq 2\mathcal{B}-1$ instead. Since $\mathcal{{A}}$ and $\mathcal{B}$ are linear expressions involving integers, these constraints are equivalent to the original ones. Finally, the constraints for which the inequality has the ``$\geq$'' direction can be handled easily using standard techniques.
	
	The connection above established that Problem \ref{problem3} can be solved in polynomial time when the number of input variables $\ell$ is fixed. To show that Problem \ref{problem2} can be solved in polynomial time as well, we will do binary search on the interval $[0,\sum_{i=1}^{\ell} \alpha_i]$ where each time a feasible solution is obtained, it will be an improvement over the previously found solution. The termination condition is when the length of the interval becomes smaller than $1/b^2$. Since the algorithm for solving Problem \ref{problem2} is invoked at most $\log{b^2\sum \alpha_i}$ times, the binary search algorithm terminates in polynomial time. The details of the algorithm can be found in Algorithm \ref{alg:binary}. 
\end{proof}

\vspace{6mm}
\begin{algorithm}[h]
	\caption{Algorithm for Problem \ref{problem2} using an algorithm for Problem \ref{problem3}.}
	\label{alg:binary}
	\begin{algorithmic}[1]
		\Require $\alpha_i \in \mathbb{Q},a,b \in \mathbb{N^*}$ 
		\Ensure $x$ 
		\State $\hat{c}:=0,\hat{d}:=\sum \alpha_i$
		\While {$\hat{d}-\hat{c}> \frac{1}{b^2}$}
		\State $t:=\frac{\hat{c}+\hat{d}}{2}$
		\If {Problem \ref{problem3} has no solution for $c=t,d=\hat{d}$}
		\State $\hat{d}:=t$
		\Else 
		\State $\hat{c}:=t$
		\EndIf
		\EndWhile
		\State return $x,k_1,\ldots,k_{\ell}$ computed for the last instance of Problem \ref{problem3}.
	\end{algorithmic}
\end{algorithm}

\bigskip
\subsection{NP-hardness of Problem 1}\label{app:problem1hard}

A natural question would be whether one could use Problem \ref{problem1} to construct a polynomial-time algorithm for finding a revenue-maximizing envy-free pricing scheme in linear multi-unit markets, for any number of buyers. In this subsection, we prove that unless P=NP, this is not possible. Note that the hardness of the problem does not imply NP-hardness of finding a revenue-maximizing envy-free pricing scheme in the market.

\begin{theorem}\label{thm:problem1hard}
	Problem \ref{problem1} is NP-hard.
\end{theorem}

To prove Theorem \ref{thm:problem1hard}, we will construct a series of problems that each could be solved using a polynomial time algorithm for the previous problem, with the last problem being the well-known \textsc{$k$-Clique} problem. Consider the following problem.

\begin{problem}\label{problem3}
	Input: $p_i,q_i \in \mathbb{N^*}$ for $i=1,2,\ldots,n$ and $a,b \in \mathbb{N^*}$
	
	Output: $\arg\min \sum_{i=1}^n \frac{p_ix \mod q_i}{q_ix}$, $x \in [a,b] \cap \mathbb{N^*}$
\end{problem}

It is not hard to see that Problem \ref{problem1} is equivalent to Problem \ref{problem3}, and it is  sufficient to consider the case when $x \in (k\Pi q_i,(k+1)\Pi q_i),k \in \textbf{Z}$, otherwise it is trivial to obtain an objective function of value $0$. Next, we will prove that the objective function satisfies a monotonicity condition when $a$ is sufficiently large.

\begin{lemma}\label{lem:rxalarge}
	Denote $r(x)=\sum \frac{p_ix \mod q_i}{q_i}$. If $a>n(\Pi q_i)^2$, then \[r(x_1) \leq r(x_2) \Leftrightarrow \frac{r(x_1)}{x_1} \leq \frac{r(x_2)}{x_2}.\]
\end{lemma}

\begin{proof}
	First, notice that following the discussion above about the interval in which we should be searching for $x$, it holds that $r(x_2) \geq r(x_1)+\frac{1}{\Pi q_i}$. This implies that
	
	\[r(x_1)x_2-r(x_2)x_1 \leq r(x_1)x_2-r(x_1)x_1 - \frac{x_1}{\Pi q_i}.\] 
	
	If $r(x_1) \leq r(x_2)$, then $r(x_1)x_2-r(x_2)x_1$ is at most $n \Pi q_i - \frac{x_1}{\Pi q_i}$ which is at most $0$, since $x_2-x_1< \Pi q_i$, $r(x) \leq n$ for any $x$ and $x_1 \geq a \geq n (\Pi q_i)^2$. 
	
	The other direction can be shown similarly.
\end{proof}

Now consider the following problem and observe that by Lemma \ref{lem:rxalarge}, it is special case of Problem \ref{problem3}.

\begin{problem}\label{problem4}
	Input: $p_i,q_i \in \mathbb{N^*}$ for $i=1,2,\ldots,n$ and $a,b \in (n\Pi q_i,(n+1)\Pi q_i)$.
	
	Output: $\arg\min \sum_{i=1}^n \frac{p_ix \mod q_i}{q_i}$, $x \in [a,b] \cap \mathbb{N^*}$
\end{problem}

By subtracting $n \Pi q_i$ from the solutions to Problem \ref{problem4}, we obtain the following equivalent problem.
\begin{problem}\label{problem5}
	Input: $p_i,q_i \in \mathbb{N^*}$ for $i=1,2\ldots,n)$ and $a,b \in (0,\Pi q_i) \cap \mathbb{N^*}$
	
	Output: $\arg\min \sum_{i=1}^n \frac{p_ix \mod q_i}{q_i}$, $x \in [a,b] \cap \mathbb{N^*}$
\end{problem}

As a special case of Problem 5, we double $n$ and set $p_i=1$ for $n$ terms and $p_i=q_i-1$ for the other $n$ terms. We obtain the following problem.

\begin{problem}\label{problem6}
	Input: $a,b \in (0,\Pi q_i) \cap \mathbb{N^*}, q_i\in \mathbb{N^*}$ for $i=1,2,\ldots,n$
	
	Output: $\arg\min \sum_{i=1}^n \frac{x \mod q_i}{q_i}+\frac{(q_i-1)x \mod q_i}{q_i},x \in [a,b] \cap \mathbb{N^*}$
	
\end{problem}

Note that for each term in the sum of Problem \ref{problem6}, the term it equals to 0 if $q_i | x$ the term equals to 1 if $q_i \nmid x$. To minimize the objective function in Problem \ref{problem6} we want to maximize the number of zero terms. Let $U=\{q_1,q_2,...,q_n\}$ denote the set of the $q_i$ and consider the following problem, which can be solved by solving Problem \ref{problem6}.

\begin{problem}\label{problem7}
	Input: $a,b \in (0,\Pi q_i) \cap \mathbb{N^*}, U$
	
	Output: $\max k$ such that $X \subseteq U, |X| = k, q_i | x ~\forall q_i \in X, x \in [a,b] \cap \mathbb{N^*}$
\end{problem}
Next we construct a special case of Problem \ref{problem6}. We first construct a base set $W$ and then set each $q_i$ to be the product of some elements in $W$.
Let \[W=\{ M,M+1,M+2,...,M+n^2-1\}, M>2^nn^4.\] 

We prove the following lemma.

\begin{lemma}
	If $1 \leq x \leq (M+n^2-1)^k$, then for any $k(k < n^2)$ different elements in $W$ there exists $x$ s.t. $x$ is divisible by them simultaneously. There does not exist $x$ s.t. $x$ is divisible by any $k+1$ different elements in $W$ simultaneously.
\end{lemma}

\begin{proof}
	The first part is obvious. For the second part, note that the greatest common divisor between any two elements in $W$ is at most $n^2$, so if $x$ is divisible by $k+1$ different elements in $W$ simultaneously, $x > \frac{M^{k+1}}{\frac{1}{2}k(k+1)n^2}$, the right side is larger than $(M+n^2-1)^k$. 
\end{proof}

For Problem \ref{problem6}, given $T<n^2$, we set $a=1,b= (M+n^2-1)^T$, and each $p_i$ to be the product of a subset $Q_i$ of $W$. For this to be a special case of Problem \ref{problem7}, it has to be that $\Pi q_i > b$. To ensure this, we first consider the case when $T \geq \sum_i |Q_i |$, where the solution can easily be seen to be $n$. For the case when $T < \sum_i |Q_i |$, setting $b=(M+n^2-1)^T$ means that $b < \Pi q_i$ since the product of any $T+1$ elements in $W$ is larger than $(M+n^2-1)^T$.  Now $x$ is divisible by at most $T$ different elements in $W$ simultaneously and we obtain the following problem as a special case.

\begin{problem}\label{problem8}
	Input: $Q_1,Q_2,...,Q_n \subseteq W, T$
	
	Output: $\max k$ such that $X \subseteq W, |X| \leq T, \sum_{i=1}^n I(Q_i \subseteq X)=k$
\end{problem}

Now consider a graph $G=(V,E)$ where $|V|=n$, and let $Q_i=\{w_{i1},w_{i2},...,w_{in}\}$, where $w_{ij} \in W$ and let $w_{ij}=w_{ji}$ if $(i,j) \in E$, $w_{ij}  \neq w_{ji}$ if $(i,j) \notin E$. Since $|W|=n^2$ there are enough elements to construct $Q_i$.  Given $k<n$, we set $T=nk-\frac{1}{2}k(k-1)$ and we obtain a special case of Problem \ref{problem8}.

\begin{problem}\label{problem9}
	On input $k$, decide whether there exists $X \in W$ s.t. $|X| \leq nk-\frac{1}{2}k(k-1), \sum_{i=1}^n I(Q_i \subseteq X)=k$
\end{problem}

Finally, we will prove that Problem \ref{problem9} is NP-hard, by a reduction to \textsc{$k$-Clique}.

\begin{lemma}
	The answer to Problem \ref{problem9} is yes if and only if Graph $G$ has a k-clique. 
\end{lemma}

\begin{proof}
	Note that the cardinality of the union of the $k$ subsets $Q_i$ equals to $nk$ minus the total number of edges in the $Q_i$'s corresponding subgraph of $G$. It is larger than $nk-\frac{1}{2}k(k-1)$ if the corresponding subgraph is not a clique, so $X$ can not cover them. 
\end{proof}

Since every problem was a special case of the previous one, this establishes Theorem \ref{thm:problem1hard}.

%\begin{theorem}
%Problem 1 is NP-hard. 
%\end{theorem}
%
%\begin{proof}
%Problem 1 $\Leftrightarrow$ problem 4 $\Rightarrow$ problem 5 $\Leftrightarrow$ problem 6 $\Rightarrow$ problem 7 $\Leftrightarrow$ problem 8 $\Rightarrow$ problem 9 $\Rightarrow$ problem 10 $\Leftrightarrow$ k-clique problem 
%\end{proof}

%\subsection*{A reduction from a clean problem}
%
%As we mentioned in the discussion, in this subsection we formulate a clean number theoretic problem which could be useful for proving that finding a revenue-maximizing price and allocation for any number of buyers is NP-hard. The problem is the following.
%
%\begin{problem}\label{problemA}
%	Given $Q_i \in \mathbb{N^*}$, output $x \in \argmax \sum_{i = 1}^{n} \frac{x \mod Q_i} {Q_i}$, $x \in [a,b] \cap \mathbb{N^*}$.
%\end{problem}  
%
%We will prove the following lemma.
%
%\begin{lemma}
%	An polynomial-time algorithm for finding a revenue-maximizing envy-free pricing scheme in linear multi-unit markets would also solve Problem \ref{problemA}.
%\end{lemma}	
%

\subsection{The General Model}\label{sec:general}

In general, the valuation of a buyer for different numbers of copies is not necessarily a linear function of that number, so each buyer has a valuation vector $v_i$, such that $v_{i,j}$ represents buyer $i$'s valuation for receiving $j$ units of the good.  
%The utility of a buyer $i$ at allocation $\vec{x}$ and price $p$ is:
\begin{equation}
u_i(p, x)=\begin{cases}
v_{i,x_i} - p \cdot x_i, & \text{if $p \cdot x_i \leq B_i$}.\\
- \infty, & \text{otherwise}.
\end{cases}
\end{equation}
The \emph{demand} of a buyer is defined as \[D_i = \{y: y \in \argmax_{x_i} u_i(p,x_i)\ \ , \ \ p\cdot y \leq B_i\},\] i.e. a set consisting of allocations that make the buyer maximally satisfied at the chosen price and respect its budget constraint. The definition of an \emph{envy-free pricing scheme} is the same as the one given in the Preliminaries section.

We will first prove that finding a social welfare-maximizing or a revenue-maximizing envy-free price and allocation is NP-hard by a reduction from \textsc{Subset-sum} and then we prove that both problems admit an FPTAS by using a variant of \textsc{Knapsack} called \textsc{Multi-choice Knapsack}. Note that since buyers have to specify their values $v_{i,j}$ for all possible number of units that they might receive, the input size to the problems here is $n$ and $m$ instead of $n$ and $\log m$, which was the input size for the linear multi-unit market. This means that an NP-hardness result for the linear case would not imply NP-hardness for the general case and an FPTAS for the general case does not imply an FPTAS for the linear case.

\begin{theorem} \label{thm:general_hard}
Given a multi-unit market with general valuations,
it is NP-hard to compute a revenue-maximizing or social welfare-maximizing envy-free price and allocation.
\end{theorem}

\begin{proof}
We construct a reduction from the NP-complete problem \textsc{Subset-Sum}, which is defined next:
\begin{quote}
\emph{Given a universe of positive integers $\mathcal{U} = \{s_1, \ldots, s_n\}$ and an integer $K$, determine whether there exists a subset $S \subseteq U$ that sums up to exactly $K$.}
\end{quote}
Given an input $\langle \mathcal{U}, K \rangle$ to \textsc{Subset-Sum}, we construct a market as follows. Let $N = \{1, \ldots, n\}$ be the set of buyers and $m = K$ the number of units.
Set the
budget of each buyer $i$ to $b_i = s_i$ and the value of $i$ for $j$ units of the good to:
\[
\label{eq1}
v_{i}(j) = 
\left\{
	\begin{array}{ll}
		s_i & \mbox{if } j \geq s_i \\
		0 & \mbox{if } j < s_i
	\end{array}
\right.
\]
We show that a revenue of at least $K$ can be obtained via an envy-free pricing in this market if and only if the \textsc{Subset-Sum} problem has a solution.

Suppose there is a price $p$ and allocation $\vec{x}$ such that $(p, \vec{x})$ is envy-free and the revenue attained is at least $K$. Since each buyer $i$ gets an optimal bundle, we have: $u_i(x_i, p) = v_i(x_i) - p \cdot x_i \geq 0$. There are three cases:
\begin{enumerate}[1).]
\item $x_i = s_i$. Then $s_i - p \cdot s_i \geq 0$, and so $p \leq 1$.
\item $x_i < s_i$. Then $v_i(x_i) = 0$, so $u_i(x_i, p) = 0 - p \cdot x_i \geq 0$, which means that either $p = 0$ or $x_i = 0$. Since $p > 0$ at any positive revenue, it must be that $x_i = 0$.
\item $x_i > s_i$. Then $v_i(x_i) = s_i$, so $u_i(x_i, p) = s_i - p \cdot x_i \geq 0$, so $p \leq s_i/x_i < 1$. Since the revenue is at least $K$ and there are exactly $K$ units, it cannot be that $p < 1$, and so this case can never occur.
\end{enumerate}
From cases $(1-3)$ it follows that each buyer $i$ gets either $s_i$ or zero units and that the price is at most $1$. Since the revenue is at least $K$ and attained from selling $K$ units, we get that in fact $p = 1$. Let $\mathcal{I} = \{i \in N \; | \; x_{i} > 0\}$ be the set of buyers that get a non-zero allocation. Then $Rev(p, \vec{x}) = \sum_{i \in \mathcal{I}} x_{i} \cdot p = \sum_{i \in \mathcal{I}} s_i = K$, which implies that $\mathcal{I}$ is a solution to the \textsc{Subset-Sum} problem.

For the other direction, let $\mathcal{I}$ be a solution to \textsc{Subset-Sum}. Set the price to $p = 1$ and the allocation $\vec{x}$ to
$x_{i} = s_i$. It can be checked that $(p, \vec{x})$ is an envy-free pricing with revenue exactly $K$.

The reduction for social welfare is almost identical and exploits the fact that the equilibrium price achieving the target social welfare must, again, be equal to $1$ (if it exists).
\end{proof}

\begin{theorem} \label{thm:general_fptas}
Given a multi-unit market with general valuations, there exists an FPTAS for the problem of computing a welfare or revenue maximizing envy-free price.
\end{theorem}
\begin{proof} (sketch)
The main idea is that for general valuations, the hardness no longer comes from guessing the price---in fact we can freely iterate over the set of candidate prices and the optimal solution is still guaranteed in this set for the same reason as that in the case of general valuations. Rather, the hardness comes from selecting the set of buyers to be allocated at a given price. 

Thus if we had a black box that produced an approximate solution efficiently for each fixed price, then the problem would be solved. This black box will be the FPTAS for the 0-1 \textsc{Multi-Choice Knapsack Problem}; there are several such algorithms, such as the one due to Lawler~\cite{Lawler79}, which for $n$ items and precision $\epsilon > 0$, runs in time $O(n \log{1/\epsilon} + 1/\epsilon^4)$.

\medskip

The 0-1 \textsc{Multi-Choice Knapsack Problem} is as follows:
\begin{quote}
We are given $m$ classes $C_1$, $C_2$, . . . , $C_m$ of items to pack in
some knapsack of capacity $C$. Each item $j \in C_i$ has a value $p_{i,j}$ and a weight $w_{i,j}$. The problem
is to choose at most one item from each class such that the value sum is maximized without the
weight sum exceeding $C$.
\end{quote}

Then our algorithm for maximizing revenue or welfare in a market with general valuations is as follows. For each price $p \in \mathcal{P}$, where $\mathcal{P}$ is the set of candidate prices, construct a 0-1 \textsc{Multi-Choice Knapsack} instance, such that the $i$th category contains an item for every quantity (i.e. \# of units) that can be sold to player $i$. 
That is, we create an element in category $C_i$ for every possible integer $y \in \{0, \ldots, m\}$ with the property that player $i$'s utility is non-negative should he be allocated exactly $y$ units. Note that some numbers of units will be missing from $C_i$, exactly at those sizes at which $i$'s utility would be negative. 
\begin{itemize}
\item For the objective of maximizing revenue, set the weight of the item corresponding to $y$ equal to $y$ and its value equal to $p \cdot y$.
\item For the objective of maximizing social welfare, set the weight of each item corresponding to $y$ equal to $y$ and its value to $v_{i,y}$.
\end{itemize}
Then set the total volume constraint to $m$ and the total value (target $k$).

Then it can be seen that a value of $K$ can be obtained in the 0-1 \textsc{Multi-Choice Knapsack} problem if and only if a target revenue (or welfare respectively) of $K$ can be obtained in the market. Note that if a knapsack solution omits taking an element from some class $C_i$, we can obtain an equivalent solution that is still feasible where the item with zero value and zero weight from class $C_i$ is included on top of all the other elements that are part of the solution.

Since we have an FPTAS for solving the 0-1 \textsc{Multi-Choice Knapsack} instance, we can just this algorithm as a subroutine $n \times m$ times, which makes the total runtime still polynomial in the market size and $1/\epsilon$.
\end{proof}

\end{document}